\documentclass[11pt,a4paper]{article}

\usepackage{amsmath,amssymb}

\usepackage{amsthm}
\usepackage{hyperref}
\usepackage{mathtools}
\usepackage{enumitem}
\usepackage{bbm}
\usepackage{url}
\usepackage{graphicx,graphics,float}
\usepackage{subfigure}
\usepackage[left=1in, right=1in, top=1in, bottom=1in, includefoot, headheight=15pt]{geometry}
\usepackage{stmaryrd}
\usepackage{enumitem}

\usepackage{color}
\definecolor{darkgreen}{rgb}{0, .5, 0}
\definecolor{darkred}{rgb}{.5, 0, 0}

\usepackage{setspace}

\theoremstyle{plain}
\newtheorem{theorem}{Theorem}[section]
\newtheorem{proposition}[theorem]{Proposition}
\newtheorem{assumption}[theorem]{Assumption} 
\newtheorem{lemma}[theorem]{Lemma} 
\newtheorem{example}[theorem]{Example}
\theoremstyle{definition} 
\newtheorem{definition}[theorem]{Definition}

\numberwithin{equation}{section}
\usepackage{multicol}

\usepackage{accents}

\usepackage{mathrsfs}

\newcommand{\E}{{\mathbb{E}}}
\providecommand{\R}{{\mathbb{R}}}

\newcommand{\dd}{{\rm d}}

\providecommand{\N}{{\mathbb N}}

\newcommand{\1}{\ensuremath{\mathbf{1}}}

\ifx\prop\undefined
\newtheorem{prop}{Proposition}[section]
\fi

\newtheorem{remark}[theorem]{Remark}

\usepackage{amsthm}
\makeatletter
\def\th@plain{%
  \thm@notefont{}
  \itshape 
}
\def\th@definition{%
  \thm@notefont{}
  \normalfont 
}
\makeatother

\def\P{{\mathbb P}}

\providecommand{\abs}[1]{\ensuremath{\left\lvert#1\right\rvert}}

\begin{document}
\singlespacing %
\title{\LARGE\bf Managing Default Contagion in Inhomogeneous Financial Networks}
\author{Nils Detering\thanks{Department of Statistics and Applied Probability, University of California, Santa Barbara, CA 93106, USA. Email: detering@pstat.ucsb.edu} \thanks{Nils Detering acknowledges financial support from the "Frankfurter Institut f\"ur Risikomanagement und Regulierung (FIRM)" and the Europlace Finance Institute.}, Thilo Meyer-Brandis\thanks{Department of Mathematics, University of Munich, Theresienstra\ss{}e 39, 80333 Munich, Germany. Emails: meyerbra@math.lmu.de, kpanagio@math.lmu.de and ritter@math.lmu.de}, Konstantinos Panagiotou\thanks{Konstantinos Panagiotou has received funding from the European Research Council, ERC Grant Agreement 772606–PTRCSP.} \footnotemark[3], Daniel Ritter\footnotemark[3]  }
\maketitle

\begin{abstract}
\noindent The aim of this paper is to quantify and manage systemic risk caused by default contagion in the interbank market. We model the market as a random directed network, where the vertices 
represent financial institutions and the weighted edges monetary exposures. Our model captures the strong
degree of heterogeneity observed in empirical data and the parameters can easily be fitted to real data sets.
Our first main result allows us to determine the impact of local shocks, where initially some banks default, to the entire system and the wider economy. Here the impact is measured by some index of total systemic importance of all eventually defaulted institutions. As a central application, we characterize resilient and non-resilient cases. In particular, for the prominent case where the network has a degree sequence without second moment, we show that a small number of initially defaulted banks can trigger a substantial default cascade. Our results complement and extend earlier findings derived in the configuration model, where the existence of a second moment of the degree distribution was assumed. As a second main contribution, paralleling regulatory discussions, we determine minimal capital requirements for financial institutions sufficient to make the network resilient to small shocks. An appealing feature of these capital requirements is that they can be determined locally by each institution without knowing the complete network structure, as they depend only on the institution's exposures to its counterparties.
\end{abstract}


\section{Introduction}\label{sec:intro}
\textbf{State of the Art}~~Already in 2003 Duffie and Singleton \cite{Duffie2003b} listed \textit{Systemic Risk} as one of five types of risk financial institutions are exposed to, and after the financial crisis in 2007 it gained major importance. From today's viewpoint, research is very diverse and parallels recent regulatory discussions that take systemic risk considerations into account, see \cite{BaselCommittee2013,Fed2015}. 
An overview of the different approaches to study systemic risk can be found e.\,g.~in~\cite{book:sr}. One important line of research addresses explicitly the network structure of the financial system, where institutions correspond to vertices in the network and edges represent dependencies among them, for example monetary exposures. Note that such exposures are more various than only usual \textit{Loans}; they can result from \textit{Securities Cross-holdings}, \textit{Derivatives} and \textit{Foreign Exchange}, see~\cite{Thurner2015}.  The  monograph \cite{Hurd2016} is an excellent reference for methods relying on network models. There, recent literature is summarized and \textit{Systemic Risk} is identified as comprising of some triggering shock event and its propagation through the system by which a major impact on the macroeconomy may occur. Moreover, \cite{Hurd2016} lists the following four main channels propagating the initial shock: \textit{Asset Correlation}, \textit{Default Contagion}, \textit{Liquidity Contagion}, and \textit{Market Illiquidity and Asset Fire Sales}.

Among the above mentioned propagation mechanisms of systemic risk, \textit{Default Contagion} probably is the one that has been studied the most in the existing literature. A prominent model was developed by Eisenberg and Noe \cite{Eisenberg2001}, 
who investigate uniqueness of a clearing vector 
for liabilities when some institutions cannot fully pay off their debt. The model was extended in various works, see for example \cite{Rogers2013,Weber2016}, where several issues, such as the assumption that default happens without additional costs or additional contagion channels, are addressed. Gai and Kapadia \cite{Gai2010} even assumed a recovery rate of zero, i.\,e.~default costs of $100\%$. This is reasonable for investigating the contagion mechanism, since processing defaults may take months or even years -- time that institutions in financial distress usually do not have. Furthermore, right after the default of an institution there is large uncertainty about its assets' value and the mark-to-market recovery rate is hence likely to be very low. As a striking example, in a bond auction for the settlement of credit default swaps written on Lehman Brothers just three weeks after its default the realized recovery rate only amounted to $8.625\%$~\cite{creditfixings}.

There are various possible approaches to design an adequate network model for the financial system. The most direct one is to work with the concrete observed structure of the financial network of interest. This is the approach used in \cite{Eisenberg2001} and its extensions. Similarly, \cite{Capponi2016} develops matrix majorization tools that allow to compare financial systems with different liability concentration in terms of the systemic loss generated. Financial networks also experience some change over time, however, and to make statements about the resilience of possible scenarios of the financial system in future it might not be advisable to simply consider today's observed network. In \cite{Chong2016} for example, the authors develop a structural default model and use a Bayesian network approach to derive formulas for the joint default and survival probability. Although not directly related to a default cascade, another structural stochastic model that describes interbank lending was proposed in \cite{Fouque2013b} and modified/extended in \cite{Carmona2014,Kley2014}.

An alternative approach, and the one being followed here, is to consider a random graph which is such that a typical sample resembles the important statistical characteristics -- for example the degree distribution -- of the real network. A strength of this approach is that it allows for the employment of  probabilistic limit theorems to asymptotically obtain analytic results for large networks in the analysis of systemic risk. Moreover, these results are then robust with respect to local changes of the network over time, as they are expressed in terms of statistical characteristics of financial networks which have been shown to be relatively stable over time, see e.\,g.~\cite{Cont2013}. 
A popular choice for such a random model is the configuration model, as pursued by Amini et al.~\cite{Cont2016,Amini2014c} for example. Among other results, 
they showed that a financial network is resilient to small initial shocks if and only if a specific measure only depending 
on the number of so-called \emph{contagious links} is negative; a debt is referred to as contagious if the bank cannot sustain the default of the corresponding counterparty. Thus, in this approach resilience is a property that can be characterized in terms of local effects only.

The resilience criterion in~\cite{Cont2016} is rather strong, but the paper makes a crucial assumption on the structure of the networks: it requires that the underlying degree distribution has a finite second moment. The reason is a technical one: without this condition the probability that the configuration model generates a simple network (i.\,e.~with no multiple edges or loops) is tiny, see also \cite{Janson2009}, and hence results that hold with high probability for the configuration model are not necessarily true when conditioning on simplicity of the graph. Real networks, however, are so 
inhomogeneous that their degree distribution does not necessarily have this property. Evidence of this is given in \cite{Boss2004,Cont2013}, where the degrees of the Austrian, respectively the Brazilian banking system, are shown to be Pareto distributed with parameter $\beta\in(2,3)$. Apart from that, a key property of real world financial networks, see for example \cite{Craig2014}, is a distinct core-periphery structure, where a few large banks are connected to many other large or small banks, but small banks are connected only to few others. In models of random graphs such structures appear only when the underlying degree distribution has no second moment~\cite{Hofstad2014vol2}. 

In the literature on random graphs there have been several proposals for constructing inhomogeneous simple random graphs. In \cite{ChungLu2002} for example the network is constructed for specified expected degrees (also called weights or fitness) of the vertices via a multiplicative specification of edge probabilities. In \cite{Bollobas2007} a more general kernel function for the probabilities is chosen.  
Whereas these examples construct undirected networks, however, financial networks clearly are directed. In \cite{Detering2015a}, the first three authors of this article have studied a directed version of the random graph in \cite{ChungLu2002};  this model will be the base for the model developed here. Also see \cite{Gandy2017} for a directed fitness model in the context of financial networks.

\vspace{1mm}
\noindent  \textbf{Contribution of this Work}~~The first aim of this paper is to develop and study a general model that can be used in various relevant situations. In particular, we will complement and extend the results in~\cite{Cont2016,Amini2014c} so that they are applicable to broader and more prominent settings, where for example the degree distributions are allowed to have an infinite 
second moment. As it turns out, deriving precise results in the more general setting with less assumptions is possible; however, we will also observe that the resilience criteria are not as simplistic as in~\cite{Cont2016,Amini2014c} and capture many global effects and interactions. Another main contribution is then the determination of sufficient minimal capital requirements for controlling systemic risk in financial networks. A striking feature of the derived management strategies is that the capital requirement  can essentially be determined from local information only, i.\,e.~from the profile of the particular bank. This is in contrast to other management/allocation rules obtained in deterministic networks that can only be specified in terms of the complete network.

Our starting point is the paper \cite{Detering2015a}, where the first three authors developed a directed random graph model that makes it possible to study contagion on random graphs whose underlying degree distribution is not necessarily required to possess a second moment,  while simultaneously preserving the simplicity of the network. Whereas financial networks are clearly weighted in the sense that exposures between banks have a monetary value, \cite{Detering2015a} is only capable of investigating the following simplified contagion mechanism: each bank is assigned an integer-valued threshold that represents the number of debtors in the network that need to default in order for the bank to default as well. It is not clear how such threshold values could be determined for an observed network since the actual sizes of loans also play an important role in the contagion mechanism for financial networks. 
Having said this, we make here  the following contributions.

\vspace{1mm}
\noindent  \textit{A Random Graph Model for Financial Networks}~~
We develop a model that combines many properties observed for real financial networks such as simplicity and directedness of the links between institutions, weighted connections and a strong degree of inhomogeneity. As a basic building block we use the model from \cite{Detering2015a}, and we enhance it with capitals on the vertices and weights on the edges. This allows us to use some results from \cite{Detering2015a} in the proofs for the enhanced model. We make asymptotic statements about the impact that initial defaults have on the system and the wider economy due to their propagation through the network via a cascade mechanism.
Compared to the existing literature, where mostly only the number of defaulted banks is considered a 
systemic risk factor, our setup allows to account for a more general systemic importance that each bank may have for the network or the real economy, e.\,g.~by providing infrastructure for the payment system or a considerable share in the lending business. This is in line with latest regulatory methods, where banks are categorized depending on their systemic importance \cite{BaselCommittee2013,Fed2015}. While our model is intended primarily to describe \textit{Default Contagion}, it can also be used to describe \textit{Illiquidity Contagion} or other percolation processes on weighted networks.

\vspace{1mm}
\noindent \textit{Resilience Criteria and Systemic Risk Capital Requirements}~
We derive explicit criteria that determine whether a financial network is resilient to contagion with respect to small shocks caused by extraordinary events, such as for example a stock market crash, natural disasters or war, that trigger the default of a few institutions. Our model specification allows us to choose these shocks in an arbitrarily correlated way, which is in line with the channel of \textit{Asset Correlation} listed above. 
Under certain conditions  we can determine the amplification factor of small initial shocks to the network. We then employ the resilience criteria to derive a formula for the banks' risk capital that is sufficient to make the system resilient to initial shocks. Due to the fact that we allow for infinite second moment of the degree sequences in our model, the derived risk capital turns out to be more restrictive than simply prohibiting contagious links as proposed in \cite{Cont2016}.
We derive a   formula for a threshold that will make the system resilient, which depends on the (expected) in-degree $w^-$ of a certain bank by a sublinear form $\alpha (w^-)^\gamma$, where $\alpha>0$ and $\gamma \in (0,1)$.
For the natural case of upper tail dependent degree sequences the threshold is sharp in the sense that barriers $\alpha_\text{c}$ and $\gamma_\text{c}$ for the values of $\alpha$ and $\gamma$ exist below which the system is non-resilient and above which it is always resilient.  We then state how such threshold values can be transformed into monetary capital requirements. This contributes to the ongoing discussion about adequate risk capital that was for example addressed 
by the Basel Committee on Banking Supervision in \cite{BaselCommittee2013} or by the Board of Governors of the Federal Reserve System in \cite{Fed2015}. 

A further important aspect in the context of risk management is the question about how to allocate the total risk to the individual institutions and the related question of deriving capital charges for them. In this context, complicated questions like `What is a fair allocation?' or `Should one impose higher systemic risk charges on institutions that in some sense are prone to transmitting systemic risk or on institutions that connect to such systemic banks and hence expose themselves to systemic risk?' arise. Approaches to these problems have for example been developed in \cite{Biagini2017,Feinstein2017}. However, to the best of our knowledge, all existing proposals of systemic risk allocations in the literature require the knowledge of the complete network. While a regulator might possess this knowledge, it seems complicated to communicate such systemic risk charges when the institutions cannot reconstruct those themselves. An  appealing feature of our approach is that each institution can compute its own systemic risk charge by basically just considering its local neighborhood, i.\,e., by knowing its counter parties.
In particular, this inherently ensures fairness in the sense that a bank's capital requirement only depends on its own business decisions. Furthermore, agents are prevented from manipulating capital charges of competitors. 
We remark that our results are in line with the analysis in \cite{Capponi2016} where it was found that the real network topology is {\it unbalancing} and as a consequence a deconcentration of exposures is desirable for stability.  

\vspace{1mm}
\noindent \textbf{Outline}~~The article is structured as follows: Section \ref{random:graph} describes our model for financial networks and states our main result about the damage caused by a default cascade in the model. In the third section, we first state sufficient criteria for resilience respectively non-resilience and then provide a rigorous proof of the formula describing the necessary risk capital of banks in a financial system. In Section \ref{simulation:study} we pursue simulation studies to support our findings. All proofs in this article are placed in Section \ref{sec:proofs}.

\section{Default Contagion on a Weighted, Directed Random Graph}\label{random:graph}
We shall present a stochastic model for a weighted, directed financial network. It will be based on the directed random graph model proposed in \cite{Detering2015a} (see Subsection \ref{ssec:special:case:threshold:model}) but complemented by edge weights. 
The main objective will be to assess the damage caused by default contagion asymptotically when the network size grows to infinity.

\subsection{Default Contagion and Systemic Importance}\label{ssec:default:contagion:systemic:importance}
We first describe the process of default contagion on a given (deterministic) financial network. If $n\in\N$ is the size of the network, we label the institutions (for simplicity called banks hereafter) by indices $i \in [n]$, where $[n]:=\{1,\dots,n \}$, and interpret them as vertices in a graph. Furthermore, $e_{i,j}\in\R_{+,0}$ describes the exposure of bank $j$ to bank $i$ and we include a directed edge of weight $e_{i,j}$ from $i$ to $j$ in the graph if $e_{i,j}>0$. We do not allow for self-loops or multiple edges between two vertices pointing in the same direction. That is, $e_{i,i}=0$ for all $i\in[n]$ and the network structure is completely determined by the exposure matrix $(e_{i,j})_{i,j\in[n]}$. Moreover, consider for each bank $i\in[n]$ its capital/equity $c_i\in\R_{+,0,\infty}:=\R_{+,0}\cup\{\infty\}$ and a value of \emph{systemic importance} $s_i\in\R_+$ which measures the potential damage caused by the default of bank~$i$. A regulating institution can specify $s_i$ 
in various ways, see for example the approaches 
developed by the \textit{Basel Committee on Banking Supervision}
\cite{BaselCommittee2013}, 
or
the \textit{Board of Governors of the Federal Reserve System} \cite{Fed2015}, 
as well as the concept of \textit{DebtRank} as introduced in \cite{Battiston2012a}. 

We call bank $i$ \emph{solvent} if $c_i>0$ and \emph{insolvent/defaulted} if $c_i=0$ (caused by some exogenous shock to the network). The set of \emph{initially defaulted} banks is hence given by $\mathcal{D}_0 = \{ i \in [n] \,:\, c_i = 0\}$. They trigger a default cascade $ \mathcal{D}_0\subseteq  \mathcal{D}_1\subseteq ...$ given by 
\begin{equation}\label{eqn:default:cascade}
\mathcal{D}_k = \Bigg\{ i \in [n] \,:\, c_i \le \sum_{j\in \mathcal{D}_{k-1}} e_{j,i} \Bigg\},
\end{equation} 
where in each step $k\geq1$ of the cascade process bank $i$ has to write off its exposures to banks that defaulted in step $k-1$ and goes bankrupt as soon as its total write-offs exceed its initial capital. The chain of default sets clearly stabilizes after at most $n-1$ steps and we call $\mathcal{D}_n=\mathcal{D}_{n-1}$ the {\em final default cluster} in the network induced by $\mathcal{D}_0$. We could easily introduce a constant recovery rate $R\in[0,1)$ to our model by multiplying exposures $e_{j,i}$ by a factor $1-R$ in \eqref{eqn:default:cascade}.

A first approach, that is often pursued in current literature, is to identify the damage caused to the financial network with the fraction $n^{-1}\vert\mathcal{D}_n\vert$. That is, damage is bearable if only few banks default as a result of the external shock event and the thereby started cascade process and it becomes the more threatening the larger the final fraction of defaulted banks $n^{-1}\vert\mathcal{D}_n\vert$ gets. In line with current regulator considerations, however, it is more realistic to consider the more general index of systemic importance of defaulted banks to really measure the damage to the economy. Instead of the size of the final default cluster $\mathcal{D}_n$, in the following we hence consider its total systemic importance $\mathcal{S}_n:=\sum_{i\in\mathcal{D}_n}s_i$ as a measure for the damage caused. Clearly, the particular  case $\mathcal{S}_n=\vert\mathcal{D}_n\vert$ is covered by setting $s_i=1$ for each $i\in[n]$.

\subsection[A Special Case: the Threshold Model]{A Special Case: the Threshold Model from \cite{Detering2015a}}\label{ssec:special:case:threshold:model}
Consider for now the special case that $e_{i,j}\in\{0,1\}$ and $c_i\in\N_{0,\infty}:=\N_0\cup\{\infty\}$. That is, whether or not a bank in the network defaults depends on the number of defaulted debtors (and the bank's individual integer-valued capital $c_i$). This particular setting, which we call the threshold model, forms the basis of the analysis in this article.  
It was studied in detail in \cite{Detering2015a} and we recall the most relevant results here.

Instead of a deterministic network structure, we describe the network as a random graph. To this end, in addition to capital $c_i$ and systemic importance $s_i$, assign to each vertex $i\in[n]$ two deterministic vertex-weights $w_i^-\in\R_+$ and $w_i^+\in\R_+$ and define the probability $p_{i,j}$ of a directed edge from vertex $i$ to vertex $j$ being present by
\begin{equation}\label{conn:prob}
p_{i,j}=\min \{1,w^+_i w^-_j/n \}\1_{i\neq j}.
\end{equation}
Further, let $X_{i,j}$ be the indicator random variable 
for the event of edge $(i,j)$ sent from vertex $i$ to vertex~$j$ being present and assume that these events are independent for all $i,j\in[n]$. The role of in-weight $w^-_i$ respectively out-weight $w^+_i$ is to determine the tendency of vertex $i\in [n]$ to have incoming respectively outgoing edges. The vertex-weights are deterministic and purely used as a mean to specify the edge probabilities. They should not be confused with the edge-weights~$e_{i,j}$. The construction of the random graph via vertex-weights resembles the one in \cite{ChungLu2002,Bollobas2007};  note, however, that our graphs are directed.

For each configuration of the network we then consider the cascade process \eqref{eqn:default:cascade} to obtain the final default cluster $\mathcal{D}_n$ and its  systemic importance $\mathcal{S}_n$. Since the network is random, these quantities are random variables. The idea in the following is to let the network grow in a regular fashion (see Assumption \ref{ass:regularity}) and to use law-of-large-numbers effects in order to derive a deterministic limit for $n^{-1}\mathcal{S}_n$.
In particular, for each network size $n\in\N$ let $\mathbf{w}^-(n)=(w_1^-(n),\ldots,w_n^-(n))$, $\mathbf{w}^+(n)=(w_1^+(n),\ldots,w_n^+(n))$, $\mathbf{s}(n)=(s_1(n),\ldots,s_n(n))$ and $\mathbf{c}(n)=(c_1(n),\ldots,c_n(n))$ be sequences of in-weights, out-weights, systemic importance values and capitals of the individual banks. We impose the following regularity conditions:
\begin{assumption}\label{ass:regularity}
For each $n\in\N$, denote the joint empirical distribution function of $\mathbf{w}^-(n)$, $\mathbf{w}^+(n)$, $\mathbf{s}(n)$ and $\mathbf{c}(n)$ by
\[ F_n(x,y,v,l)=n^{-1}\sum_{i\in[n]}\1\{w_i^-(n)\leq x,w_i^+(n)\leq y,s_i(n)\leq v,c_i(n)\leq l\},\quad  (x,y,v,l)\in\R_+^3\times\N_{0,\infty}, \]
and let $(W_n^-,W_n^+,S_n,C_n)$ a random vector distributed according to $F_n$. We assume that:
\begin{enumerate}
\item\label{ass:regularity:1} \textbf{Convergence in distribution:} There exists an $x_0 > 0$ and a distribution~$F$ on $\R_+^3\times\N_{0,\infty}$ such that $F(x,y,v,l)=0$ for all $x,y\leq x_0$, and such that at all points $(x,y,v,l)$ for which $F_l(x,y,v):=F(x,y,v,l)$ is continuous it holds $\lim_{n\to\infty}F_n(x,y,v,l)=F(x,y,v,l)$. Denote by $(W^-,W^+,S,C)$ a random vector distributed according to $F$.  
\item \textbf{Convergence of average weights and systemic importance:} $W^-$, $W^+$ and $S$ are integrable and $\E[W_n^-]\to\E[W^-]$, $\E[W_n^+]\to\E[W^+]$, $\E[S_n]\to\E[S]$ as $n\to\infty$.
\end{enumerate}
\end{assumption}

\noindent This assumption is of a technical nature and concerned with the behavior of the network parameters as the size of the network tends to infinity. For practical purposes one can think of Assumption \ref{ass:regularity} ensuring that the limiting network keeps the observed parameter distribution of some real network we want to investigate. In particular, the expected weights are assumed to stay finite. In \cite{Detering2015a} it was shown  that Assumption \ref{ass:regularity} implies $D_i^-\sim\mathrm{Poi}(w_i^-\E[W^+])$ respectively $D_i^+\sim\mathrm{Poi}(w_i^+\E[W^-])$ in the limit $n\to\infty$, where $D_i^-$ and $D_i^+$ denote the random in- respectively out-degree of vertex $i$ with weights $(w_i^-,w_i^+)$ and $\mathrm{Poi}(x)$ denotes a Poisson distributed random variable with parameter $x$. Conversely, one can show that for an observed network topology, i.\,e.~given in- and out-degrees, maximum likelihood estimators of the in- and out-weights are approximately given by the in- and out-degrees (normalized by some global factor). That is, morally one can think of the in- respectively out-weight of a vertex as its in- respectively out-degree.

Further, note that we do not assume $W^-$ or $W^+$ to have finite second moment, that is, the  model captures networks without a second moment condition on their degree sequences. In particular, choosing $W^-$ and $W^+$ power law distributed with parameters $\beta^-$ respectively $\beta^+$ results in power law distributions for the degrees $D^-$ and $D^+$ with the very same parameters. This allows to calibrate our model parameters to observed empirical in- and out-degree sequences. As we will see in Subsections \ref{ssec:threshold:requirements} and \ref{ssec:capital:requirements}, these power law parameters carry the most important information about the network when it comes to determining  capital requirements ensuring resilience to contagion.

We want to study the default process (\ref{eqn:default:cascade}), in particular the final default cluster~${\cal D}_n$ and its total systemic importance ${\cal S}_n$. To this end consider the following heuristics. Suppose that  $\zeta\in[0,\E[W^+]]$ denotes the total out-weight of finally defaulted banks divided by~$n$. Then in the limit $n\to\infty$ and assuming that each bank defaults independently we obtain for any bank $i\in[n]$ that the number of finally defaulted neighbors in the network should be close to a random variable distributed like $\mathrm{Poi}(w_i^-\zeta)$. Bank $i$ is thus finally defaulted itself if and only if $\mathrm{Poi}(w_i^-\zeta)\geq c_i$. Summing over all banks in the network we therefore derive the  identity
\[ \E[W^+\psi_C(W^-\zeta)] = \zeta, \]
where
\[ \psi_r(x) := \P\left(\mathrm{Poi}(x)\geq r\right) = \begin{cases}\sum_{j\geq r}e^{-x}x^j/j!,&0\leq r<\infty,\\0,&r=\infty.\end{cases} \]
Moreover, summing up the systemic importance values, the final damage caused by defaulted banks should be given by $\E[S\psi_C(W^-\zeta)]$.
Motivated by these heuristics consider now the function
\[ f(z;(W^-,W^+,C)):=\E\left[W^+\psi_C(W^-z)\right]-z. \]
By the dominated convergence theorem, $f(z;(W^-,W^+,C))$ is continuous and has a smallest root $\hat{z}\in[0,\E[W^+]]$. Furthermore, let
\[ d(z;(W^-,W^+,C)) := \E[W^-W^+\phi_C(W^-z)]-1, \]
the weak derivative of $f$ (see Lemma \ref{lem:f:continuous}), where
\[
	\phi_r(x) := \P\left(\mathrm{Poi}(x)=r-1\right)\1_{r\geq1}.
\]
A sequence of events $(E_n)_{n\in\N}$ shall hold with high probability (w.\,h.\,p.) if $\P(E_n)\to1$, as $n\to\infty$. The following theorem for the threshold model will be used in the proofs of our main results in this article. 
\begin{theorem}[adapted from {\cite[Theorem 7.2]{Detering2015a}}]\label{thm:threshold:model}
Consider a sequence of financial systems satisfying Assumption \ref{ass:regularity} and let $\hat{z}$ be the smallest positive root of $f(z ;(W^-,W^+,C))$. Then the following holds:
\begin{enumerate}
\item For all $\epsilon>0$, with high probability
$ n^{-1}\mathcal{S}_n \geq \E\left[S\psi_C(W^-\hat{z})\right] - \epsilon.$
\item If $d(z;(W^-,W^+,C))$ is bounded from above by some constant $\kappa<0$ on a neighborhood of $\hat{z}$, then
\[ n^{-1}\mathcal{S}_n \xrightarrow{p} \E\left[S\psi_C(W^-\hat{z})\right],\quad\text{as }n\to\infty. \]
\end{enumerate}
\end{theorem}
\noindent The theorem thus allows us to either bound from below or explicitly compute the final damage~$\mathcal{S}_n$ for $n\to\infty$. In case that $f$ is differentiable, the assumption on $d$ 
in the second part of the theorem ensures that $f'(\hat{z};(W^-,W^+,C))<0$ and that the first zero of $f$ determines the end of the process.

\subsection{The Exposure Model}\label{ssec:exposure:model}

Based on the threshold model from the previous subsection, we will now construct our weighted, directed random graph model for financial systems. In particular, we assign weights to the edges, and each institution is equipped with a capital and a systemic importance.

We model the occurrence of edges by the random matrix $X=X(n)=(X_{i,j})_{i,j\in[n]}$ from Subsection \ref{ssec:special:case:threshold:model} and we assign to each pair $(i,j)\in[n]^2$ with $i\neq j$ a random variable $E_{i,j}>0$ representing $j$'s \emph{possible} exposure to $i$ such that $E=E(n)=(E_{i,j})_{i,j\in[n]}$ is independent of $X$. The random exposure of $j$ to $i$, where $j \neq i$, is then given by $e_{i,j}=X_{i,j}E_{i,j}$, and we set $e_{i,i} = 0$. We make one more assumption in order to facilitate our computations. We assume that for each bank $j$ the list of possible exposures $(E_{i,j})_{i \in [n]\setminus\{j\}}$ is an exchangeable sequence of random variables. That is, for each $j\in[n]$ and each permutation $\pi$ of  $[n]\backslash\{j\}$ 
\[ \left(E_{1,j},\ldots,E_{j-1,j},E_{j+1,j},\ldots,E_{n,j}\right) \stackrel{d}{=} \left(E_{\pi(1),j},\ldots,E_{\pi(j-1),j},E_{\pi(j+1),j},\ldots,E_{\pi(n),j}\right). \]
This is the same as taking an arbitrary sequence of random variables $(\tilde{E}_{i,j})_{i \in [n]\setminus\{j\}}$ and transforming them into a list of exposures $(E_{i,j})_{i \in [n]\setminus\{j\}}$ by setting $E_{i,j}=\tilde{E}_{\pi(i),j}$ for some independent and uniformly random permutation $\pi$. We remark that the requirement of exchangeable exposures is a typical assumption made in the literature, see for example \cite{Cont2016}. Note, however, that in this setting the distribution of the exposure size $E_{i,j}$ only depends on the creditor bank $j$ and not on the debtor bank $i$, which is a criticizable assumption for example in strongly pronounced core/periphery networks where also the exposures might exhibit stronger heterogeneity. The relaxation of this assumption is technically more demanding and beyond the scope of this article, 
but we refer to \cite{Detering2018} where this problem is tackled.

The final ingredient in our model is to assign to each bank $i\in[n]$ a possibly stochastic capital value $c_i\in\R_{+,0,\infty}$ and a deterministic systemic importance value $s_i\in\R_+$. Using \eqref{eqn:default:cascade} we can then again determine the random final default cluster $\mathcal{D}_n$ and its systemic importance $\mathcal{S}_n$. It is the aim of the following subsection to derive results about convergence and deterministic bounds similar as in Subsection \ref{ssec:special:case:threshold:model} for the threshold model. In Table \ref{tab:parameters} we summarize all important parameters in the exposure model and compare them to the observed quantities in a real financial network.
\begin{table}
\begin{center}
\begin{tabular}{|c|c|}
\hline\textbf{Observed Network} & \textbf{Exposure Model}\\\hline
capital $c_i\in\R_{+,0}$, & capital $c_i\in L^0(\R_{+,0,\infty})$,\\
systemic importance $s_i\in\R_+$, & systemic importance $s_i\in\R_+$,\\
in-degree $d_i^-\in\N_0$, & in-weight $w_i^-\in\R_+$,\\
out-degree $d_i^+\in\N_0$, & out-weight $w_i^+\in\R_+$,\\
&  (edge probability $p_{i,j}=\min\{1,w_i^+w_j^-/n\}\1_{i\neq j}$)\\
exposure sequence $(e_{i,j})_{j\in[n]\backslash\{i\}}\subset\R_{+,0}$ & exchangeable sequence of possible edge weights\\
&$(E_{i,j})_{j\in[n]\backslash\{i\}}\subset L^0(\R_+)$\\\hline
\end{tabular}
\caption{Comparison of observed quantities in a financial network and the model parameters in the exposure model}\label{tab:parameters}
\end{center}
\end{table}

\subsection{Asymptotic Results for Default Contagion in the Exposure Model}\label{ssec:asymptotic:results:exposure:model}
The setting in the exposure model is more complex than in the threshold model since we cannot decide if a bank defaults only based on the number of its neighbors that default: \eqref{eqn:default:cascade} asserts that this also depends on the actual exposures between the banks. However, one crucial assumption that we made is that these exposures are exchangeable, so intuitively it should make no difference which neighbors of a given bank default, but just their actual number. 

To formalize this intuitive argument, define for each bank $i\in [n]$ the random threshold value
\begin{equation}\label{ex1:perc:thres}
\tau_i(n):= 
\inf \Bigg\{ s \in \{0\}\cup[n-1] \,:\, \sum_{\ell\leq s} E_{\rho_i(\ell),i} \geq c_i \Bigg\},
\quad \text{where}\quad
\rho_i(\ell):=\ell+\1_{\ell\geq i},
\end{equation}
with the usual convention $\inf\emptyset:=\infty$. That is, $\tau_i$ is allowed to take the value $\infty$ if capital $c_i$ is larger than the sum of all possible exposures. In this case,  $i$ can never default. The use of the enumeration $\rho_i$ becomes necessary in (\ref{ex1:perc:thres}) since we want to spare $i$ in this ordering. The value $\tau_i$ then determines the \textit{hypothetical} default threshold of $i$, assuming that $i$'s neighbors default in the order of their natural index given by $\rho_i$ and that all edges $(j,i)$, $1\leq j\leq \rho_i(\tau_i)$, $i\neq j$, are present in the graph. We denote the hypothetical threshold sequence by $
\pmb{\tau}(n)=(\tau_1(n),\ldots,\tau_n(n))$. 
The thresholds are only hypothetical, since not all of the first $\rho_i(\tau_i)$ exposures/edges must be present in the graph and the vertices do usually not default in their natural order. However, we know that the exposures are exchangeable, so all these simplifications should have no effect; it will turn out in the proof of Theorem~\ref{thm:asymp:1} that indeed the value $\tau_i$ captures the actual dynamics, and the qualitative characteristics of the contagion process in the exposure model are the same as in the threshold model with capital sequence $\pmb{\tau}(n)$.

As an equivalent of Assumption \ref{ass:regularity} for the threshold model we impose the following regularity conditions.
\begin{assumption}\label{vertex:assump}
For each $n\in\N$, denote the random joint empirical distribution function of $\mathbf{w}^-(n)$, $\mathbf{w}^+(n)$, $\mathbf{s}(n)$ and $\pmb{\tau}(n)$ by
\[ G_n(x,y,v,l)=n^{-1}\sum_{i\in[n]}\1\{w_i^-(n)\leq x,w_i^+(n)\leq y,s_i(n)\leq v,\tau_i(n)\leq l\},\hspace{8pt} (x,y,v,l)\in\R_+^3\times\N_{0,\infty}. \]
Then we assume that:
\begin{enumerate}
\item \textbf{Almost sure convergence in distribution:} There exists $x_0>0$ and a distribution $G$ on $\R_+^3\times\N_{0,\infty}$ such that $G(x,y,v,l)=0$ for all $x,y\leq x_0$, and such that at all points $(x,y,v,l)$ for which $G_l(x,y,v):=G(x,y,v,l)$ is continuous in $(x,y,v)$, it holds almost surely $\lim_{n\to\infty}G_n(x,y,v,l)=G(x,y,v,l)$. Denote by $(W^-,W^+,S,T)$ a random vector distributed according to $G$.
\item \textbf{Convergence of average weights and systemic importance:} 
$W^-$, $W^+$ and $S$ are integrable and $\int_{\R_+^3\times\N_{0,\infty}} x\,\dd G_n(x,y,v,l) \to \E[W^-]$, $\int_{\R_+^3\times\N_{0,\infty}} y\,\dd G_n(x,y,v,l) \to \E[W^+]$ as well as $\int_{\R_+^3\times\N_{0,\infty}} v\,\dd G_n(x,y,v,l) \to \E[S]$ as $n\to\infty$.
\end{enumerate}
\end{assumption}
To ensure that Assumption \ref{vertex:assump}~holds, the following regularity property must be satisfied: the joint empirical distribution of the in- and out-weights, the systemic importance and the hypothetical threshold values must eventually stabilize. See Subsection \ref{ssec:examples} for general examples of systems satisfying Assumption \ref{vertex:assump}.

For the remainder of this subsection, we consider a sequence of financial systems denoted as\linebreak $(\mathbf{w}^-(n),\mathbf{w}^+(n),\mathbf{s}(n),E(n),\mathbf{c}(n))$ and satisfying Assumption \ref{vertex:assump}. In particular, we denote by $(W^-,W^+,S,T)$ a random vector distributed according to the limiting distribution $G$ from Assumption \ref{vertex:assump}. We assume that the financial systems have experienced an external shock such that a positive fraction of banks have capital zero. In the notation from above this means $\P(T=0)>0$.  
Hence we are in a situation in which a default cascade is about to happen, and our aim is to study the size of the final default cluster ${\cal D}_n$ and its 
systemic importance ${\cal S}_n$. Our main result shows, in complete analogy to the threshold model, that as the network size gets large, the quantity $n^{-1}\mathcal{S}_n$ converges to a deterministic value which we can pin down. To this end, we denote
\[ f(z;(W^-,W^+,T)):=\E\left[W^+\psi_T(W^-z)\right]-z, \]
where as in Subsection \ref{ssec:special:case:threshold:model}
\[ \psi_r(x) := \P\left(\mathrm{Poi}(x)\geq r\right) = \begin{cases}\sum_{j\geq r}e^{-x}x^j/j!,&0\leq r<\infty,\\0,&r=\infty,\end{cases} \]
and
\[ d(z;(W^-,W^+,T)) := \E[W^-W^+\phi_T(W^-z)]-1, \]
where again as in Subsection \ref{ssec:special:case:threshold:model}
\[ \phi_r(x) := \P\left(\mathrm{Poi}(x)=r-1\right)\1_{r\geq1}. \]
Whenever $(W^-,W^+,T)$ is clear from the context, we abbreviate $f(z;(W^-,W^+,T))$ by $f(z)$ and $d(z;(W^-,W^+,T))$ by $d(z)$. The following lemma summarizes some properties of $f$ and $d$. See Subsection \ref{ssec:proofs:2} for the proof.
\begin{lemma}\label{lem:f:continuous}
The function $f(z)$ is continuous on $[0,\infty)$ and admits the following representation:
\begin{equation}\label{eqn:integral:representation}
f(z) = \E[W^+\1_{\{T=0\}}] + \int_0^z d(\xi)\dd\xi
\end{equation}
If $\P(T=0)>0$, then $f(z)$ has a strictly positive root $\hat{z}$.
\end{lemma}
In particular, $d(z)$ is the weak derivative of $f(z)$ and if $d(z)$ is continuous on some interval $I\subset[0,\infty)$, then $f(z)$ is continuously differentiable on $I$ with derivative $d(z)$. With this fact at hand we derive the following result about $\mathcal{S}_n$. It resembles Theorem \ref{thm:threshold:model} for the threshold model and indeed in the proof (see Subsection \ref{ssec:proofs:2}) we make use of this result. However, due to the hypothetical nature of the threshold sequence $\pmb{\tau}(n)$ this application is not straight-forward and requires considerable effort.
\begin{theorem}\label{thm:asymp:1}
Under Assumption \ref{vertex:assump}, suppose $\P(T=0)>0$ and let $\hat{z}$ be the smallest positive root of $f(z
)$. If the weak derivative $d(z)$ of $f(z)$ is bounded from above by some constant $\kappa<0$ on a neighborhood of $\hat{z}$, then
\[ n^{-1}\mathcal{S}_n \xrightarrow{p} \E [S \psi_T(W^-\hat{z}) ], \quad \text{ as } n\rightarrow \infty. \]
\end{theorem}
\noindent Two remarks are in order. First, if $f(z)$ is continuously differentiable on a neighborhood of $\hat{z}$ with $f'(\hat{z})<0$ (i.\,e.~$\hat{z}$ is stable), then Theorem \ref{thm:asymp:1} is applicable. This is a standard assumption in current literature. In \cite{Cont2016} for instance, the authors assume degree sequences of finite second moment. In this case, it is straightforward to show that $f(z)$ is continuously differentiable. Secondly, without the assumption of stableness, $n^{-1}\mathcal{S}_n$ does not converge to a deterministic limit in general (see \cite{Janson2012} for a comparable result in a much simpler setting). However, in the following theorem we are still able to state asymptotic bounds rather than an exact limiting value. We believe that the derived bounds are sharp in the sense that they cannot be improved without further assumptions. Proving this, however, is beyond the scope of this article.

\begin{theorem}\label{thm:asymp:2}
Under Assumption \ref{vertex:assump}, suppose $\P(T=0)>0$ and let $\hat{z}$ be the smallest positive root of $f(z)$. Further, let $z^*$ be the smallest value of $z>0$ at which $f(z)$ crosses zero,
\[ z^* := \inf\left\{z>0\,:\,f(z)<0\right\}. \]
Then the following holds:
\begin{enumerate}
\item\label{thm:asymp:2:1} For all $\epsilon>0$, with high probability $n^{-1}\mathcal{S}_n \geq  \E [S \psi_T(W^-\hat{z}) ]-\epsilon$.
\item\label{thm:asymp:2:2} If further $d(z)$ is continuous on some neighborhood of $z^*$, then for all $\epsilon>0$ with high probability $n^{-1}\mathcal{S}_n \leq  \E [S \psi_T(W^-z^*) ]+\epsilon$.
In particular, if $\hat{z}=z^*$, then
\[ n^{-1}\mathcal{S}_n \xrightarrow{p} \E [S \psi_T(W^-\hat{z}) ], \quad \text{ as } n\rightarrow \infty. \]
\end{enumerate}
\end{theorem}

\noindent See Subsection \ref{ssec:proofs:2} for the proof. Note that under the assumptions of Theorem \ref{thm:asymp:2} it holds that $f(0)>0$. Together with continuity of $f$ from Lemma \ref{lem:f:continuous} this implies that $z^*\geq\hat{z}$. In general, however, it is possible that $z^*>\hat{z}$, for example if $f$ has a local minimum at $\hat{z}$.

\subsection{Examples for Financial Systems Satisfying Assumption \ref{vertex:assump}}\label{ssec:examples}
In this subsection we demonstrate how our model can be applied to describe many different scenarios. In Example \ref{example:1}, we describe a financial system where exposures are i.\,i.\,d.~
in certain buckets.

\begin{example}\label{example:1}
Let $({\bf w^-}(n), {\bf w^+}(n), {\bf s}(n))$ be a triple consisting of in-weight, out-weight and systemic importance sequences such that their joint empirical distribution and the first moments of the marginals converge. Choose some partition of $[0,\infty)^3$ into countably many Borel sets $D_k$, $k\in \N$. Further, denote $\mathcal{W}_k:=\{i\in[n]\,:\, (w_i^-,w_i^+,s_i)\in D_k\}$. Let the distributions of $E_{j,i}$, $j\in[n]$, and~$c_i$ be equal for  vertices $i \in \mathcal{W}_k$ and assume them all to be independent. From this assumption it follows that for $n> \ell\in\N_0$ the 
distribution of $Y_i^\ell:=\1_{\{\tau_i\leq \ell\}}$ only depends on the set  $\mathcal{W}_k\ni i$. Denote by $q_k^\ell:=\P(Y_i^\ell=1)=\mathbb{E} [Y_i^\ell]$ the probability that $i\in\mathcal{W}_k$ has threshold less or equal $\ell$.

We show that Assumption~\ref{vertex:assump} is satisfied. For this let $(x,y,v)\in \mathbb{R}_+^3$ and define
\[ \mathcal{W}^{(x,y,v)}:=\{i \in [n] \,:\, w^-_i \leq x, w^+_i \leq y, s_i\leq v \}. \]
Our assumptions guarantee that  $n^{-1}|\mathcal{W}^{(x,y,v)}\cap\mathcal{W}_k|$ stabilizes as $n \to \infty$ for all $k$. 
Since $Y^\ell_i$ has exponential moments, by the strong law of large numbers it follows that a.s.
\begin{align*}
 \lim_{n\rightarrow \infty }n^{-1}\sum_{i \in \mathcal{W}^{(x,y,v)}\cap\mathcal{W}_k} Y_i^\ell &= 
\lim_{n\rightarrow \infty }n^{-1} \abs{\mathcal{W}^{(x,y,v)}\cap\mathcal{W}_k}q_k^\ell =: G (x,y,v,\ell,k) 
\end{align*}
exists for every $k\in \mathbb{N}$. Summing over all the sets $\mathcal{W}_k$ it follows that the network described by $({\bf w^-}(n), {\bf w^+}(n), {\bf s}(n))$ and $\pmb{\tau}(n)$ fulfills Assumption \ref{vertex:assump} with limiting distribution $G (x,y,v,\ell):= \sum_{k\in \mathbb{N}} G (x,y,v,\ell,k) $. 
\end{example}

\noindent As 
observed in a similar setting in \cite{Cont2016} already, the independence of the exposure random variables 
can be weakened.
\begin{example}
Similar as above, assume that vertices are partitioned into $K$ classes $\mathcal{W}_1,\dots,\mathcal{W}_K$  with vertices with the same marginal distributions of the capitals and exposures. The sets may depend on the network size $n$ but we shall assume that $\lim_{n\rightarrow \infty} n^{-1} \abs{\mathcal{W}_{k}} =:\lambda (k)$, $k\in[K]$, i.\,e.~the fraction of vertices of a given class stabilizes. 
For each   $k \in [K]$ we are given generating sequences $\{c_\ell^k\}_{\ell\in\N}$ and $\{E_\ell^k\}_{\ell\in\N}$ of random variables in $\R_+$ resp.~$\R_+\backslash\{0\}$. Further, we assume that $\{c_\ell^k\}_{\ell\in\N}$ and $\{E_\ell^k\}_{\ell\in\N}$ are infinite exchangeable systems and independent of each other
(see for example \cite{Aldous1985} for the definition of (infinite) exchangeability). For each network size $n$, assign   to every vertex $i\in\mathcal{W}_k$ a capital from $\{c_\ell^k\}_{\ell\in\N}$ and $n-1$ exposures from $\{E_\ell^k\}_{\ell\in\N}$ according to any deterministic rule such that no capital or exposure is used more than once. 
Define then the threshold value as in (\ref{ex1:perc:thres}) and for a fixed $m\in\N$ the random variable $Y^m_{k,i}:=\1_{\{\tau_i = m\} }$ that determines whether vertex $i$ has threshold value $m$. Observe that for $n\geq m+1 $ every vertex $i\in [n]$ has more than $m$ exposures and the distribution of $Y^m_{k,i}$ is thus independent of $n$.
Let $\beta_k (1),\dots , \beta_k (\abs{\mathcal{W}_k} )$ the indices of the vertices in $\mathcal{W}_k$. By construction then  
\[ \mathcal{L} ( Y^m_{k,\beta_k (1)}, \dots ,  Y^m_{k,\beta_k (\abs{\mathcal{W}_k} } ) = \mathcal{L} ( Y^m_{k,\sigma_k (\beta_k (1))}, \dots ,  Y^m_{k,\sigma_k ( \beta_k (\abs{\mathcal{W}_k} ) }) \]
for all permutations $\sigma_k$, that is, for each $k\in [K]$, the random variables $\{Y^m_{k,i}\}$ build an exchangeable system. 
Since for fixed $n$ the sequence $\{Y^m_{k,i} \}_{i\in \mathcal{W}_k}$ is just the restriction to a finite subset of variables of an infinite exchangeable system for $\abs{\mathcal{W}_k} \rightarrow \infty$ it converges in law to an infinite exchangeable system.
This implies that the system of random variables $(Y^m_{k,i})_{k\in [K],i\in \mathcal{W}_k}$ forms a multi-exchangeable system (see \cite{Graham2008} for a definition).
Define the empirical measure by
\[ \Lambda^m_k := \frac{1}{\abs{\mathcal{W}_k}} \sum_{i=1}^{\abs{\mathcal{W}_k}} \delta_{Y^m_{k,\beta_k (i)}} \]
for each $k \in [K]$. By \cite[Thm.~2]{Graham2008} convergence in distribution of $\{Y^m_{k,i} \}_{i\in \mathcal{W}_k}$ implies convergence in distribution of the empirical measure sequence $(\Lambda_k)_{k\in [K]}$, without any assumptions on the dependency structure across classes. Since the above considerations apply for all $m\in \mathbb{N}$, convergence in distribution of the empirical measure sequence 
\[ \Lambda_k := n^{-1} \sum_{i \in [n]} \delta_{\tau_i} \]
follows for all $k \in [K]$. By the Skorohod Coupling Theorem \cite[Thm. 4.30]{Kallenberg2001}, there exists a probability space with random elements $\{ \tilde{\Lambda}_k \}_{k \in [K]}$ distributed as $\{ \Lambda_k \}_{k \in [K]}$ such that $\{ \tilde{\Lambda}_k \}_{k \in [K]}$ converges almost surely as required.
\end{example}

\section{Resilient Networks and Systemic Capital Requirements}\label{sec:resilience}

In the previous section we quantified the default propagation in financial networks after an initial external shock. The aim of the present section is to develop general and easy-to-use criteria that allow us to determine how systemically risky a network is \emph{prior} to a shock event. More specifically, we consider a sequence of financial systems $(\mathbf{w}^-(n),\mathbf{w}^+(n),\mathbf{s}(n),E(n),\mathbf{c}(n))$ satisfying Assumption \ref{vertex:assump} with $\P(T>0)=1$, that is, initially there are (asymptotically) no defaults. Given this setup, we then apply some small -- possibly random -- shock to the capitals only; we call this an \emph{ex post} shock. In this context, a \textit{resilient}, systemically unrisky, network should only experience minor damage
, whereas in \textit{non-resilient}, systemically risky, networks even a small shock can cause huge harm to the whole system. An advantage over static models such as the Eisenberg-Noe model \cite{Eisenberg2001} is that we can assess stability already for an unshocked system. Further, here we will show that whether a financial network is   resilient or non-resilient only depends on the distributions of $W^-$, $W^+$ and $T$. These have been shown to be relatively stable over time even if locally the network might change noticeably.  
\subsection{Resilience Criteria for Unshocked Networks}
In order to incorporate such small random shocks into our model, we introduce a sequence $\mathbf{m}(n)=(m_1(n),\ldots,m_n(n))$ of binary marks $m_i\in\{0,1\}$ to $(\mathbf{w}^-(n),\mathbf{w}^+(n),\mathbf{s}(n),E(n),\mathbf{c}(n))$, where $m_i=0$ denotes that bank $i$ defaults ex post due to some shock event and hence loses all its capital to start the cascade process. Otherwise, the capital distribution stays the same. 
We extend Assumption \ref{vertex:assump} such that there exists a distribution $\overline{G}$ such that  
\[ \overline{G}_n(x,y,v,l,k) = n^{-1}\sum_{i\in[n]}\1\{w_i^-(n)\leq x,w_i^+(n)\leq y, s_i\leq v, \tau_i(n)\leq l, m_i(n)\leq k\} \]
converges almost surely at all continuity points $(x,y,v)$ of $\overline{G}_{l,k}(x,y,v):=\overline{G}(x,y,v,l,k)$ and denote by $(W^-,W^+,S,T,M)$ a random vector distributed according to 
$\overline{G}$. We assume that $\P(T=0)=0$, but $\P(M=0)>0$ such that indeed $M$ causes ex post defaults in an  {initially} unshocked system.

We want to declare 
a financial system to be 
non-resilient to initial shocks if even very small shocks $M$ can cause significant damage $\mathcal{S}_n^M$, measured by the total systemic importance of the defaulted banks $\mathcal{D}_n^M$ at the end of the contagion process triggered by $M$. 
\begin{definition}\label{def:non:resilience}
A financial system is said to be \emph{non-resilient} if there exists a constant $\Delta>0$ such that for each ex post default $M$ with $\P(M=0)>0$ with high probability
\[ n^{-1}\mathcal{S}_n^M \geq \Delta.
\]
\end{definition}

The following theorem states a sufficient criterion for a system to be non-resilient. 
\begin{theorem}[Non-resilience Criterion]\label{prop:nonres}
Under Assumption \ref{vertex:assump} suppose that $\P(T=0)=0$ and that there exists $z_0>0$ such that
\begin{equation}\label{condition:nonres}
f(z)>0,\quad\text{for all }0<z<z_0.
\end{equation}
Then for all $M$ with $\P(M=0)>0$  with high probability
\[ n^{-1}\mathcal{S}_n^M \geq \E\left[S\psi_T(W^-z_0)\right] . \]
In particular, if $\E[S\1_{\{T<\infty\}}]>0$, then the system is non-resilient.
\end{theorem}

The proof of Theorem~\ref{prop:nonres} follows from Part \ref{thm:asymp:2:1} of Theorem \ref{thm:asymp:2} and arguments analogue to the ones used in \cite[Thm.~7.3]{Detering2015a} and is thus omitted here. We can interpret Theorem \ref{prop:nonres} as follows. If a financial network satisfies condition (\ref{condition:nonres}), then no matter how small the fraction of banks which are driven into bankruptcy by an external shock event, after the cascade process of defaults always a damage larger than the constant $\E\left[S\psi_T(W^-z_0)\right]$ is caused to the system. In the reasonable case that $\E[S\1_{\{T<\infty\}}]>0$, this lower bound for the damage is strictly positive and the system is hence non-resilient according to Definition \ref{def:non:resilience}. In particular, by choosing $s_i=1$ for all $i\in[n]$ and hence $S\equiv 1$, we derive that the final default fraction $n^{-1}\vert\mathcal{D}_n^M\vert$ is lower bounded by the constant $\E\left[\psi_T(W^-z_0)\right]$, which is positive (unless $\P(T=\infty)=1$).

\begin{figure}[t]
    \hfill\subfigure[]{\includegraphics[width=0.4\textwidth]{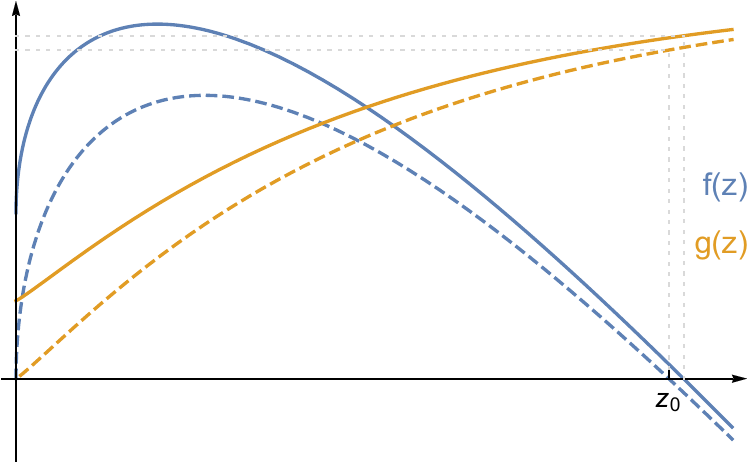}\label{fig:condition:nonres}}
    \hfill\subfigure[]{\includegraphics[width=0.4\textwidth]{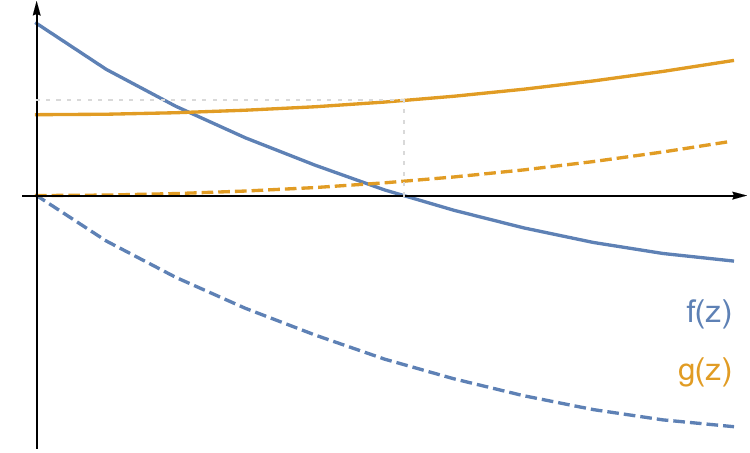}\label{fig:condition:res}}\hfill
\caption{Examples of functions $f(z)=\E[W^+\psi_T(W^-z)]-z$ (blue) satisfying conditions (\ref{condition:nonres}) (a) respectively (\ref{condition:res:2}) (b). Orange: the function $g(z)=\E[S\psi_T(W^-z)]$. Dashed: the unshocked functions. Solid: the shocked functions.}\label{fig:conditions:nonres:res}\hfill
\vspace{-8mm}
\end{figure}

Condition (\ref{condition:nonres}) is an assumption on $f$ which is illustrated in Figure \ref{fig:condition:nonres}. Whereas for $\P(M=0)=0$ the first non-negative root of the function is zero, any howsoever small increase in $\P(M=0)$, and hence upwards shift of $f(z)$, makes the first root jump above $z_0$ and causes default of a set of size larger than $n\E[\psi_T(W^-z_0)]$ and systemic damage larger than $n\E[S\psi_T(W^-z_0)]$.

If on the other hand the function $f(z;(W^-,W^+,T))$ is such behaved that the first positive root $\hat{z}^M$ of $f(z;(W^-,W^+,TM))$ tends to zero as $\P(M=0)$ becomes smaller, one can expect that also the final default cluster $\mathcal{D}_n^M$ and its systemic importance $\mathcal{S}_n^M$ vanish and the system can hence be regarded as \textit{resilient} to small shocks. See Figure \ref{fig:condition:res} for an exemplary illustration. This intuition is formalized in the following.
\begin{definition}\label{def:resilience}
A financial system is said to be \emph{resilient} if for each $\epsilon>0$ there exists $\delta>0$ such that
\[ n^{-1}\mathcal{S}_n^M \leq \epsilon\quad\text{w.\,h.\,p.~for all }M\text{ such that }\P(M=0)<\delta. \]
\end{definition}
In words this means that the final damage to the system $\mathcal{S}_n^M$ can be controlled by the initial default fraction $\P(M=0)$. Theorem \ref{prop:res} is then the analogue of \cite[Thm.~7.4]{Detering2015a} 
transferred to our exposure model.

\begin{theorem}[Resilience Criterion]\label{prop:res}
Under Assumption \ref{vertex:assump} suppose $\P(T=0)=0$ and that there exists $z_0>0$ such that
\begin{equation}\label{condition:res:2}
d(z)<0,\quad\text{for all }0<z<z_0.
\end{equation}
Then for any sequence of ex post defaults $\{M_i\}_{i\in\N}$ with $\lim_{i\to\infty}\P(M_i=0)=0$, it follows that for any $\epsilon>0$, there exists $i_\epsilon$ such that
\[ n^{-1}\mathcal{S}_n^{M_i} \leq \epsilon \quad \text{w.\,h.\,p. for all } i\geq i_\epsilon. \]
In particular, the system is resilient.
\end{theorem}

\noindent Theorem \ref{prop:res} states that the total systemic importance of all finally defaulted banks tends to zero as the initial default fraction tends to zero, which is exactly our definition of resilience. However, it makes no statement about the rate of convergence. If we assume not only that $d(z) < 0$ for $z$ small enough but even $\limsup_{z\to0+}d(z) < 0$, then we derive the following result concerning convergence speed. See Subsection \ref{ssec:proofs:3} for the proof.

\begin{proposition}\label{prop:convergence:speed}
Under Assumption \ref{vertex:assump} suppose $\P(T=0)=0$ and
\[ \kappa:=\limsup_{z\to0+}d(z)<0\quad\text{and}\quad \kappa_S:=\limsup_{z\to0+}\E\left[W^-S\phi_T(W^-z)\right]<\infty. \]
Then for any sequence $\{M_i\}_{i\in\N}$ with $\lim_{i\to\infty}\P(M_i=0)=0$, it follows that
\[ \text{w.\,h.\,p.} \quad n^{-1}\mathcal{S}_n^{M_i} \leq \E[S\1_{\{M_i=0\}}] - \kappa^{-1}\kappa_S\E[W^+\1_{\{M_i=0\}}] + o(\E[W^+\1_{\{M_i=0\}}]). \]
If $f(z)$ and $\E[S\psi_T(W^-z)]$ are continuously differentiable from the right at $z=0$ with derivatives $\kappa<0$ and $\kappa_S<\infty$, then 
\[ n^{-1}\mathcal{S}_n^{M_i} \xrightarrow{p} \E[S\1_{\{M_i=0\}}] - \kappa^{-1}\kappa_S\E[W^+\1_{\{M_i=0\}}] + o(\E[W^+\1_{\{M_i=0\}}]). \]
\end{proposition}
\noindent In particular, if $\{M_i\}_{i\in\N}$ is independent of $W^+$ and $S$, then
\[ n^{-1}\mathcal{S}_n^{M_i} \leq \P(M_i=0)\left(\E[S]-\kappa^{-1}\kappa_S\E[W^+]\right) + o(\P(M_i=0)) = \mathcal{O}(\P(M_i=0)) \]
and $1-\kappa^{-1}\kappa_S\E[W^+]/\E[S]$ can be regarded as the maximal amplification factor of the systemic importance of initially defaulted banks $\E[S\1_{\{M_i=0\}}]=\P(M_i=0)\E[S]$. If further $S\equiv1$ and $W^-W^+$ is integrable, above result is the analogon to \cite[Corollary 20]{Amini2014c}:
\[ n^{-1}\vert\mathcal{D}_n^{M_i}\vert \xrightarrow{p} \P(M_i=0)\left(1+\frac{\E[W^+]\E[W^-\1_{\{T=1\}}]}{1-\E[W^-W^+\1_{\{T=1\}}]}\right) + o(\P(M_i=0)) \]

\noindent Both Theorem \ref{prop:res} and Proposition \ref{prop:convergence:speed} are concerned with the behavior of the weak derivative $d(z)$ of $f(z)$ near $z=0$. The following criterion which rather focuses on the behavior of $f(z)$ near $z=0$ will turn out to be useful later. See Subsection \ref{ssec:proofs:3} for the proof.
\begin{theorem}\label{thm:cont:res}
Under Assumption \ref{vertex:assump} suppose $\P(T=0)=0$, $d(z)$ is continuous on $(0,z_0)$ for some $z_0>0$ and
\begin{equation}\label{thm:cont:res:ass}
\inf\left\{z>0\,:\,f(z)<0\right\} = 0.
\end{equation} 
Then for any sequence $\{M_i\}_{i\in\N}$ with $\lim_{i\to\infty}\P(M_i=0)=0$, it follows that for any $\epsilon>0$, there exists $i_\epsilon$ such that
\[ n^{-1}\mathcal{S}_n^{M_i} \leq \epsilon \quad \text{w.\,h.\,p. for all } i\geq i_\epsilon. \]
In particular, the system is resilient.
\end{theorem}
\noindent Assumption (\ref{thm:cont:res:ass}) describes that $f(z)$ becomes negative immediately after $z=0$. It is in some sense the opposite of assumption (\ref{condition:nonres}) and ensures that the roots $z^*_i$ of the shocked systems tend to zero as the shock size $\P(M_i=0)$ shrinks to zero.

\subsection{Systemic Threshold Requirements}\label{ssec:threshold:requirements}
A natural problem that is of great interest to regulators is to identify capital requirements for the individual banks which can be determined from observable quantities of the network and that are sufficient to make the network resilient to external shocks. Observable quantities are the in- and out-degrees $(d_i^-)_{i\in[n]}$ respectively $(d_i^+)_{i\in[n]}$, which serve as estimators of the in- and out-weights $(w_i^-)_{i\in[n]}$, respectively $(w_i^+)_{i\in[n]}$, and interbank exposures. In this subsection we will first focus on identifying threshold requirements in the threshold model (see Subsection \ref{ssec:special:case:threshold:model})
that guarantee resilience. In the next subsection, we then discuss how to translate the threshold requirements into systemic capital requirements in the exposure model.

More precisely, in this subsection we seek threshold requirements for bank $i$ of the form $\tau_i=\tau(w_i^-)$, where $\tau:\R_+\to\N$ is a non-decreasing function. Such a functional form has the interpretation that the threshold (capital) requirement of a bank only depends on its risk of defaulting due to default of debtors (exposure risk). In contrast, if bank $i$'s threshold (capital) requirement $\tau(w_i^-,w_i^+)$ was also depending on the out-weight $w_i^+$, this would also take possible defaults caused by bank $i$ into account. This risk management policy would not be in line with traditional risk management techniques and would be more difficult to communicate to banks.

Note, in particular, the changed nature of the threshold values $\tau_i$. While, so far, the capitals $c_i$ (which equal $\tau_i$ in the threshold model) were exogenous quantities in our model, we now aim to determine them from the in-weights $w_i^-$ under the constraint of building a resilient network, hence making them endogenous quantities.

To investigate resilience of a financial system implementing threshold requirements given by $\tau$ we want to use the resilience criteria from the previous subsection. In particular we have to ensure that Assumption \ref{vertex:assump} is satisfied for the functional thresholds $\tau_i=\tau(w_i^-)$, i.\,e.~we need
\[ \lim_{n\to\infty}n^{-1}\sum_{i\in[n]}\1\{w_i^-(n)\leq x,w_i^+(n)\leq y,s_i(n)\leq v, \tau(w_i^-(n))\leq l\} = G(x,y,v,l) \]
for some distribution $G:\R_+^3\times\N_{0,\infty}\to[0,1]$ and all points $(x,y,v,l)\in\R_+^3\times\N_{0,\infty}$ for which $G_l(x,y,v):=G(x,y,v,l)$ is continuous. Note that depending on the choice of $\tau$ and $G$ for the limiting random vector $(W^-,W^+,S,T)\sim G$ it does not necessarily hold that $\P(T=\tau(W^-))=1$. This is because $W^-$ could have positive mass at some point of discontinuity of $\tau$ and it would then be important whether the in-weight distributions converge from below or from above. Instead one easily derives that $\P(\accentset{\circ}{\tau}(W^-)\leq T\leq \bar{\tau}(W^-))=1$ where $\accentset{\circ}{\tau}(w):=\lim_{\epsilon\to0+}\tau((1-\epsilon)w)$ and $\bar{\tau}(w):=\lim_{\epsilon\to0+}\tau((1+\epsilon)w)$ are the left-continuous resp.~right-continuous modifications of $\tau$. If, however, $\tau$ only admits discontinuities at $\tilde{w}\in\R_+$ such that $\P(W^-=\tilde{w})=0$, then in fact $\P(T=\tau(W^-))=1$ and we will assume this from now on. All our results on resilience and non-resilience in the following can easily be extended for the functions $\accentset{\circ}{\tau}$ resp.~$\bar{\tau}$.
\begin{assumption}\label{ass:tau}
Consider sequences $\mathbf{w}^-(n)$, $\mathbf{w}^+(n)$ and $\mathbf{s}(n)$ of in-weights, out-weights and systemic importance values such that the empirical random vector $(W_n^-,W_n^+,S_n)$ converges in distribution to a random vector $(W^-,W^+,S)$. Moreover, let $\tau:\R_+\to\N_0$ be a non-decreasing function and assume that its points of discontinuity are all null-sets of $W^-$. In particular, letting $\tau_i(n)=\tau(w_i^-(n))$, $i\in[n]$, Assumption \ref{vertex:assump} is satisfied and it holds $T=\tau(W^-)$ a.\,s.
\end{assumption}
Empirical studies of financial networks such as \cite{Boss2004,Cont2013} show that degrees follow Pareto distributions (at least in the tail). We denote in the following $X\sim\mathrm{Par}(\beta,x_\text{min})$, $\beta>1$, $x_\text{min}>0$, if the random variable $X$ has density
\[ f_X(x)=(\beta-1)x_\text{min}^{\beta-1}x^{-\beta}\1_{x\geq x_\text{min}}. \]
As mentioned before, to reproduce Pareto distributed degrees in our model we need to choose the weights Pareto distributed as well. Hence let $W^-\sim\mathrm{Par}(\beta^-,w_\text{min}^-)$ and $W^+\sim\mathrm{Par}(\beta^+,w_\text{min}^+)$, where $\beta^->2$, $\beta^+>2$, $w_\text{min}^->0$ and $w_\text{min}^+>0$. In particular, any choice of an increasing function $\tau$ will satisfy Assumption \ref{ass:tau}. The main result of this subsection then identifies a criterion for function $\tau$ that ensures resilience of the financial network. See Subsection \ref{ssec:proofs:3} for the proof.
\begin{theorem}\label{threshold:res}
Suppose Assumption \ref{ass:tau} for a non-decreasing function $\tau:\R_+\to\N\backslash\{0,1\}$ such that for each bank $i\in[n]$ the threshold value $\tau_i$ depends on in-weight $w_i^-$ by $\tau_i=\tau(w_i^-)$. Moreover, assume for the limiting weight distributions that $W^-\sim\mathrm{Par}(\beta^-,w_\text{min}^-)$ respectively $W^+\sim\mathrm{Par}(\beta^+,w_\text{min}^+)$, $\beta^-,\beta^+>2$, $w_\text{min}^-,w_\text{min}^+>0$. Set
\[ \gamma_\text{\normalfont c}:=2+\frac{\beta^--1}{\beta^+-1}-\beta^-\quad\text{and}\quad \alpha_\text{\normalfont c}:=\frac{\beta^+-1}{\beta^+-2}w_\text{\normalfont min}^+\left(w_\text{\normalfont min}^-\right)^{1-\gamma_\text{\normalfont c}}. \]
Then the system is resilient if one of the following holds:
\begin{enumerate}
\item \label{threshold:res:1} $\gamma_\text{\normalfont c}<0$,
\item \label{threshold:res:2} $\gamma_\text{\normalfont c}=0$ and $\liminf_{w\to\infty}\tau(w)>\alpha_\text{\normalfont c}+1$,
\item \label{threshold:res:3} $\gamma_\text{\normalfont c}>0$ and $\liminf_{w\to\infty}w^{-\gamma_\text{\normalfont c}}\tau(w)>\alpha_\text{\normalfont c}$.
\end{enumerate}
\end{theorem}

\noindent The theorem identifies different criteria for $\tau$ depending on the quantity $\gamma_\text{c}$ and hence the values of $\beta^-$ and $\beta^+$. Since $\beta^->2$ and $\beta^+>2$, we note that always $\gamma_\text{c}<1$. That is, also in part \ref{threshold:res:3} of the theorem it is possible to choose a sublinear threshold function $\tau$ that ensures resilience. On the other hand, even the constant threshold function $\tau(w)=2$ for all $w\in\R_+$ ensures resilience by part \ref{threshold:res:1} whenever $\gamma_\text{c}<0$. This is in particular the case if $\beta^->3$ and $\beta^+>3$, that is, if $W^-$ and $W^+$ both admit finite second moments. This is in line with the results from \cite{Cont2016}. In addition, the theorem makes statements about the {important case $\min \{ \beta^-,\beta^+\}<3$ }. Such parameters were observed on real markets for example in \cite{Cont2013}. In these cases, all $\gamma_\text{c}<0$, $\gamma_\text{c}=0$ or $\gamma_\text{c}>0$ are possible and only the exact values of $\beta^-$ and $\beta^+$ determine the condition for resilience.

\begin{remark}\label{rem:contagious:links}
In Theorem \ref{threshold:res} we make the assumption $\tau(w)\geq2$. In other words, each bank must at least be capable of sustaining the default of its largest debtor. This requirement has already been implemented in an even stricter form in the \textit{Supervisory framework for measuring and controlling large exposures} by the \textit{Basel Committee on Banking Supervision} from 2014 which has become applicable from January 2019 \cite{BaselCommittee2014}. While being economically sensible, the assumption ``$\tau(w)\geq2$'' is actually not necessary in order to derive analytical results regarding (non-)resilience in the following sense. For the case of $\gamma_\text{c}<0$ it is actually enough to postulate $\E\left[W^-W^+\1_{\{\tau(W^-)=1\}}\right]<1$ in order to ensure resilience. Also in the case of $\gamma_\text{c}\geq0$, it suffices to adjust $\alpha_\text{c}$ for a factor $(1-\E[W^-W^+\1_{\{\tau(W^-)=1\}}])^{-1}$, whenever $\E\left[W^-W^+\1_{\{\tau(W^-)=1\}}\right]<1$. Both results follow from a simple modification of our proof and we omit the details.
\end{remark}

\noindent Note that Theorem \ref{threshold:res} is formulated with assumptions on the marginal distributions of $W^-$ and $W^+$ only. Hence, the result is robust with respect to the dependency structure of the weights, i.\,e.~the resilience criteria are sufficient for all dependency structures. As Theorem \ref{thm:functional:nonres} will show, in the case of comonotone weights, the values of $\gamma_\text{c}$ and $\alpha_\text{c}$ are sharp. Also in the case of upper tail dependent weights (a reasonable assumption for real financial networks) the value of $\gamma_c$ is sharp. By $W^-$ and $W^+$ being upper tail dependent we mean that
\[ \lambda:=\liminf_{p\to0}\P(F_{W^+}(W^+)>1-p \mid F_{W^-}(W^-)>1-p) > 0. \]
If even
\[ \Lambda(x) := \lim_{p\to0}\P(F_{W^+}(W^+)>1-xp \mid F_{W^-}(W^-)>1-p) \]
exists for all $x\geq0$, we are able to determine explicitly sharp thresholds $\alpha_\text{c}(\Lambda)$ given by
\[ \alpha_\text{c}(\Lambda) := w_\text{min}^+(w_\text{min}^-)^{1-\gamma_\text{c}}\int_0^\infty\Lambda\left(x^{1-\beta^+}\right)\dd x. \]
For comonotone dependence ($\Lambda(x)=1\wedge x$), $\alpha_\text{c}(\Lambda)$ coincides with $\alpha_\text{c}$ from Theorem \ref{threshold:res}.

\begin{theorem}\label{thm:functional:nonres}
Consider the same situation as in Theorem \ref{threshold:res}. If $\gamma_c>0$, the following holds:
\begin{enumerate}
\item \label{thm:functional:nonres:1} If $\limsup_{w\to\infty}w^{-\gamma_\text{\normalfont c}}\tau(w) < \lambda \frac{\beta^+-2}{\beta^+-1} \alpha_\text{c}$, then the system is non-resilient.
\item \label{thm:functional:nonres:2} If $\Lambda(x)$ exists for each $x\geq0$ and $\limsup_{w\to\infty}w^{-\gamma_\text{\normalfont c}}\tau(w) < \alpha_\text{\normalfont c}(\Lambda)$, then the system is non-resilient. If $\liminf_{w\to\infty}w^{-\gamma_\text{\normalfont c}}\tau(w) > \alpha_\text{\normalfont c}(\Lambda)$, then the system is resilient.
\end{enumerate}
\end{theorem}

\noindent See Subsection \ref{ssec:proofs:3} for the proof. In part \ref{thm:functional:nonres:2} of the theorem, we characterize threshold functions $\tau$ that are asymptotically smaller respectively larger than $\alpha_\text{c}(\Lambda)w^{\gamma_\text{c}}$. In the proof we calculate the derivative of $f(z)$ at $z=0$ in order to show non-resilience ($f'(0)>0$) respectively resilience ($f'(0)<0$). If $\tau(w)$ asymptotically behaves like $\alpha_\text{c}(\Lambda)w^{\gamma_\text{c}}$, we obtain $f'(0)=0$ and hence both (\ref{condition:nonres}) and (\ref{thm:cont:res:ass}) are possible (not simultaneously). In this case, the exact form of $\tau$ and not only its asymptotics are important to decide whether the system is resilient or non-resilient.

\begin{remark}
If the weights $W^-$ and $W^+$ are not upper tail dependent, the conditions from Theorem \ref{threshold:res} are generally too strict. If their dependency is such that $\E[W^+(W^-)^{1-\gamma}]<\infty$ for some $\gamma\in(0,\gamma_\text{c}]$ for example, then $\liminf_{w\to\infty}w^{-\gamma}\tau(w)>0$ is already a sufficient criterion for resilience of the system. This can easily be derived from line (\ref{eqn:expectation:upper:part}) in the proof of Theorem \ref{threshold:res}.
\end{remark}

\noindent Theorems \ref{threshold:res} and \ref{thm:functional:nonres} both describe financial systems whose weights are given by Pareto distributed random variables. While such random variables model the tails of empirical degree distributions very well, typically for small weights there is a non-negligible discrepancy. However, the proofs of Theorems \ref{threshold:res} and \ref{thm:functional:nonres} show that it is in fact only the tail that determines (non-)resilience of a financial system. Therefore, assume in the following that there exist constants $K^-,K^+\in(0,\infty)$ and $\beta^-,\beta^+>2$ such that
\begin{equation}\label{eqn:Pareto:type}
1-F_{W^\pm}(w) \leq \left(\frac{w}{K^\pm}\right)^{1-\beta^\pm}
\end{equation}
for $w$ large enough. That is, the tail distributions of $W^-$ and $W^+$ are bounded by the powers $1-\beta^-$ resp.~$1-\beta^+$. Then the following version of Theorem \ref{threshold:res} holds.
\begin{theorem}\label{thm:Pareto:type}
Suppose Assumption \ref{ass:tau} for a non-decreasing function $\tau:\R_+\to\N\backslash\{0,1\}$ such that for each bank $i\in[n]$ the threshold value $\tau_i$ depends on in-weight $w_i^-$ by $\tau_i=\tau(w_i^-)$. Moreover, let the distribution functions of $W^-$ and $W^+$ satisfy \eqref{eqn:Pareto:type}. For $\gamma_\text{\normalfont c}$ defined as before, the system is resilient if one of following holds:
\begin{enumerate}
\item \label{thm:Pareto:type:1} $\gamma_\text{\normalfont c}<0$,
\item \label{thm:Pareto:type:2} $\gamma_c=0$ and $\liminf_{w\to\infty}\tau(w)>\frac{\beta^+-1}{\beta^+-2}K^+K^-+1$,
\item \label{thm:Pareto:type:3} $\gamma_c>0$ and $\liminf_{w\to\infty}w^{-\gamma_c}\tau(w)>\frac{\beta^+-1}{\beta^+-2}K^+(K^-)^{1-\gamma_c}$.
\end{enumerate}
\end{theorem}
\noindent See Subsection \ref{ssec:proofs:3} for the proof. Note that by the same means also Theorem \ref{thm:functional:nonres} can be generalized. For non-resilience the inequality in \eqref{eqn:Pareto:type} needs to be inverted such that it describes a lower bound on the tail of the distributions.

\subsection{Systemic Capital Requirements}\label{ssec:capital:requirements}
In this subsection we translate the threshold requirements from Theorem \ref{threshold:res} to capital requirements in the exposure model. That is, we state explicit amounts of capital each bank has to be able to procure in stress scenarios in order for the system to be resilient. As for the threshold requirements, it is important to note that each bank can compute its capital requirements on its own by just knowing its local neighborhood in the network. Further, a bank's capital requirement only depends on the default risk the bank exposes itself to and not on the default risk the bank poses to other banks. Proposition \ref{prop:robust:capital:requirements} states a straightforward robust way to translate threshold requirements into sufficient capital requirements. In general, it might lead to capital requirements that are too high and hence unnecessarily reduce interbank lending and liquidity, however. Thus, we further provide Theorem \ref{cor:threshold:res} below, which accurately determines capital requirements under a certain regularity assumption on the exposure lists.

\begin{proposition}\label{prop:robust:capital:requirements}
Suppose Assumption \ref{ass:tau} for a non-decreasing function $\tau:\R_+\to\N\backslash\{0,1\}$ and limiting weights $W^-\sim\mathrm{Par}(\beta^-,w_\text{min}^-)$ resp.~$W^+\sim\mathrm{Par}(\beta^+,w_\text{min}^+)$ with $\beta^-,\beta^+>2$, $w_\text{min}^-,w_\text{min}^+>0$. Further, assume that $\liminf_{w\to\infty}\tau(w)>\alpha_\text{\normalfont c}+1$ if $\gamma_\text{\normalfont c}=0$ respectively $\liminf_{w\to\infty}w^{-\gamma_\text{\normalfont c}}\tau(w)>\alpha_\text{\normalfont c}$ if $\gamma_\text{\normalfont c}>0$, where the quantities $\gamma_\text{\normalfont c}$ and $\alpha_\text{\normalfont c}$ are as in Theorem \ref{threshold:res}. Then the system is resilient if 
\[ c_i > \max\Bigg\{\sum_{j\in J}E_{j,i}~\Bigg\vert~J\subset[n], \vert J\vert=\tau(w_i^-)-1\Bigg\} \quad\text{almost surely for all }i\in[n], \]
i.\,e.~capital $c_i$ of bank $i\in[n]$ is larger than the sum of the $\tau(w_i^-)-1$ largest exposures of $i$.
\end{proposition}

\noindent See Subsection \ref{ssec:proofs:3} for the proof. Analogously, a robust translation of Theorems \ref{thm:functional:nonres} and \ref{thm:Pareto:type} to the exposure model is possible.

Proposition \ref{prop:robust:capital:requirements} requires each bank $i$ to be able to cope with default of its $\tau(w_i^-)$ largest exposures. But as we have seen in the proof of Theorem \ref{threshold:res}, only the thresholds and hence the capitals of large banks in the network matter for resilience. For large banks with many exposures on the other hand one can expect an averaging effect of the exposure sizes to occur if they are not too irregular. Hence, one can presume that in this case multiplying threshold values from the threshold model by average exposure sizes for each bank leads to the same resilience characteristics. We formalize this in Theorem \ref{cor:threshold:res} under Assumption \ref{ass:exposures} on the exposure sequences. This assumption is motivated by the following reasoning:

For each bank $i$, let $\{E_{j,i}\}_{j\in\N\backslash\{i\}}$ be a sequence of i.\,i.\,d.~positive random variables. Let $\mu_i:=\E[E_{\rho_i(1),i}]<\infty$ be their mutual expectation and denote $S_k^i:=\sum_{j=1}^k E_{\rho_i(j),i}$. If there is some $t>1$ such that $\E\left[\vert E_{\rho_i(1),i}\vert^t\right]<\infty$, then by the Baum-Katz-Theorem from \cite{Baum1965} for all $\epsilon>0$,
\begin{equation}\label{eqn:Baum:Katz:1}
k^{t-1}\P\left(S_{k}^i\geq (1+\epsilon)k\mu_i\right) \to 0,\quad \text{as }k\to\infty,
\end{equation}
and for all $x>1$,
\begin{equation}\label{eqn:Baum:Katz:2}
k^{tx-1}\P\left(S_{k}^i\geq\epsilon \mu_i k^x\right) \to 0,\quad \text{as }k\to\infty.
\end{equation}

\begin{assumption}\label{ass:exposures}
Motivated by the above, we assume that for each bank $i\in[n]$ with exposure list $\{E_{j,i}\}_{j\in\N\backslash\{i\}}$ of mutual mean $\mu_i$, we can find $t>1$ such that the convergences in (\ref{eqn:Baum:Katz:1}) and (\ref{eqn:Baum:Katz:2}) hold. Moreover, we assume them to be uniform for $i\in[n]$ (but not necessarily for $\epsilon$ or $x$).
\end{assumption}

\noindent Assumption \ref{ass:exposures} ensures a certain regularity of the exposures without bounding their mean.

\begin{theorem}\label{cor:threshold:res}
Suppose Assumption \ref{ass:tau} for a non-decreasing function $\tau:\R_+\to\N\backslash\{0,1\}$ and such that $W^-\sim\mathrm{Par}(\beta^-,w_\text{min}^-)$ and $W^+\sim\mathrm{Par}(\beta^+,w_\text{min}^+)$ with $\beta^-,\beta^+>2$, $w_\text{min}^-,w_\text{min}^+>0$. The quantities $\gamma_\text{\normalfont c}$ and $\alpha_\text{\normalfont c}$ shall be defined as in Theorem \ref{threshold:res}. Further, assume $c_i>\max_{j\in[n]\backslash\{i\}}E_{j,i}$ almost surely for all $i\in[n]$. Then the following holds:
\begin{enumerate}
\item \label{cor:threshold:res:1} If $\gamma_\text{\normalfont c}<0$, then the system is always resilient.
\end{enumerate}
Now further assume that the exposure lists $\{E_{j,i}\}_{j\in\N\backslash\{i\}}$, $i\in\N$, satisfy Assumption \ref{ass:exposures} for some $t>1$. Then the system is resilient if one of the following holds:
\begin{enumerate}
\setcounter{enumi}{1}
\item \label{cor:threshold:res:2} $\gamma_\text{\normalfont c}=0$ and there exist some $\gamma>0$ 
such that $\liminf_{w\to\infty}w^{-\gamma}\tau(w)>0$ and for all $i\in[n]$, $c_i\geq \tau(w_i^-)\mu_i$ almost surely.
\item \label{cor:threshold:res:3} $\gamma_\text{\normalfont c}>0$, $\liminf_{w\to\infty}w^{-\gamma_\text{\normalfont c}}\tau(w)>\alpha_\text{\normalfont c}$ and for all $i\in[n]$, $c_i\geq \tau(w_i^-)\mu_i$ almost surely.
\end{enumerate}
\end{theorem}

\noindent See Subsection \ref{ssec:proofs:3} for the proof. Theorem \ref{cor:threshold:res} provides the banks with a formula that is easy to use and only requires the regulator to announce $\alpha_\text{c}$ and $\gamma_\text{c}$. Resilient capital requirements are then determined according to average exposure size $\mu_i$ and number of exposures $d_i^-\sim w_i^-$. Since the average exposure size $\mu_i$ is proportional to $(d_i^-)^{-1}$ while the factor $\alpha_\text{c}(d_i^-)^{\gamma_\text{c}}$ is sublinear in $d_i^-$, in particular a deconcentration of loans is favorable for the banks to reduce systemic risk charges.

\begin{remark}
Theorem \ref{cor:threshold:res} extends Theorem \ref{threshold:res} to the exposure model under Assumption \ref{ass:exposures} for the exposure sequences. By the same means, also Theorems \ref{thm:functional:nonres} and \ref{thm:Pareto:type} can be extended.
\end{remark}

\section{Simulation Study}\label{simulation:study}
All previous chapters have been formulated in the limit as the number of banks $n$ tends to $\infty$ and the fraction of initially defaulted banks $p$ tends to 0. It is hence reasonable to investigate whether the results are good approximations also for real finite networks with only a few thousand institutions and experience a shock of a positive fraction of banks. Since specific transactions between banks are not disclosed to the public there is no data basis for us to investigate real networks, however. Instead, we discuss our findings by simulating networks.

\subsection{Simulations for the Threshold Model}
For our simulations we make use of the observation in \cite{Cont2013} that the empirical in- and out-degrees as well as the exposure sizes in the Brazilian banking network are power law distributed. For November 2008, the authors of \cite{Cont2013} estimated the power law exponents $\beta^-=2.132$ and \linebreak$\beta^+=2.8861$ for the degree sequences and $\xi=2.5277$ for the exposures. In our weight-based model, these degree distributions are obtained by choosing in- and out-weights power law distributed with exponents $\beta^-$ and $\beta^+$ as well. In addition to this, we assume them to be comonotone and Pareto distributed with minimal weights $w_\text{min}^-=w_\text{min}^+=1$. 

In a first simulation, we consider a threshold model with above weight parameters and assume absence of contagious links, that is, we set $\tau_i=2$ for all $i\in[n]$. In order to start the cascade process, we assume initial default of $p=1\%$ uniformly chosen banks in the network. We then simulate the default process for $n\in\{100k\,:\, k\in[100]\}$ and $100$ different configurations of the random network for each $n$. The results for the final fraction of defaulted banks are plotted in Figure \ref{fig:Convergence}. As can be seen from Figure \ref{fig:TheoreticalFraction}, the theoretical value of the final default fraction as $n$ tends to infinity can be determined to be approximately $84.54\%$. This value is drawn as a red line in Figure \ref{fig:Convergence}. Already for small $n$, most of the simulations yield results that are close to this theoretical value and the networks can hence be understood as being non-resilient. As $n$ grows to $10^4$ the final fractions become even more precise. In particular, there is not a single resilient sample anymore for $n\geq500$.

\begin{figure}[t]
    \hfill\subfigure[]{\includegraphics[width=0.4\textwidth]{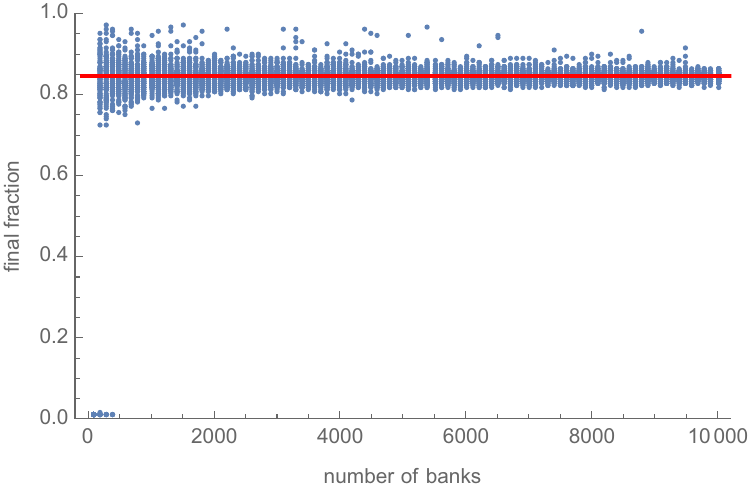}\label{fig:Convergence}}
    \hfill\subfigure[]{\includegraphics[width=0.4\textwidth]{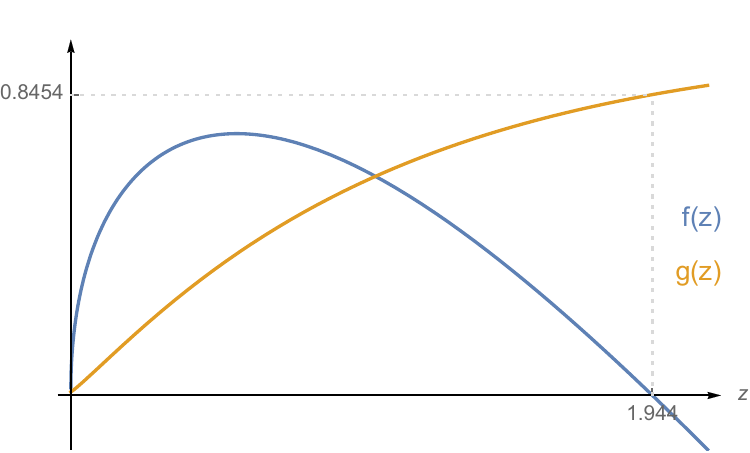}\label{fig:TheoreticalFraction}}\hfill
\caption{(a) Convergence of the final fraction of defaulted banks in the threshold model for networks of finite size. (b) Determination of the theoretical final default fraction in the threshold model for networks whose sizes grow to infinity and with $p=1\%$ initial defaults and constant threshold $2$. Blue: $f(z)=(1-p)\E[W^+\psi_2(W^-z)]+p\E[W^+]-z$ with root $\hat{z}\approx 1.94433$. Orange: $g(z)=(1-p)\E[\psi_2(W^-z)]+p$ with $g(\hat{z})\approx 0.845434$.}\label{fig:Convergence:TheoreticalFraction}
\end{figure}
\begin{figure}[t]
	\hfill\subfigure[]{\includegraphics[width=0.4\textwidth]{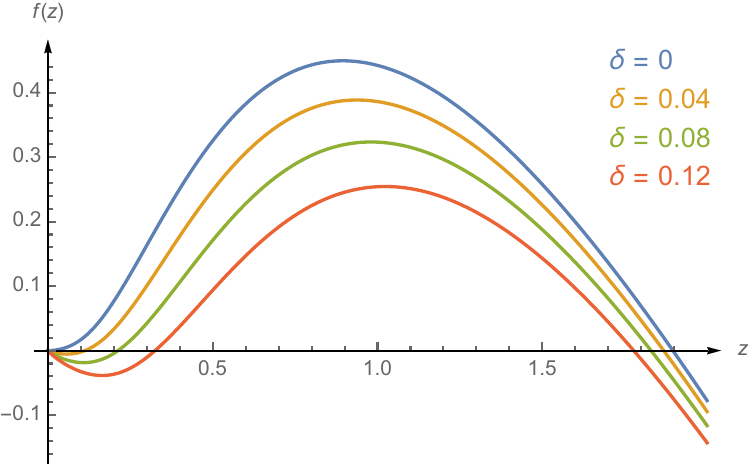}\label{fig:influenceOfDelta}}
    \hfill\subfigure[]{\includegraphics[width=0.4\textwidth]{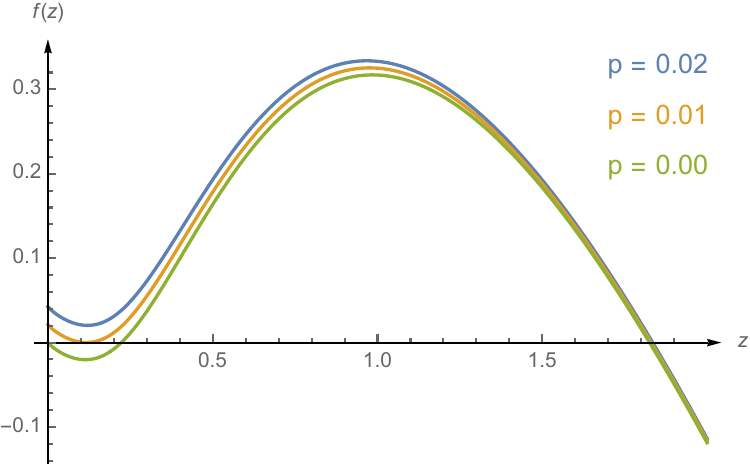}\label{fig:influenceOfP}}\hfill
\caption{(a) Influence of $\delta$ on the shape of $f(z)=\E[W^+\psi_T(W^-z)]-z$ with capital requirements $\tau_i=\max\{2,\lfloor (\alpha_\text{c}(1+\delta) (w_i^-)^{\gamma_\text{c}(1+\delta)}\rfloor\}$. (b) Influence of $p$ on the shape of function $f(z)=(1-p)\E[W^+\psi_T(W^-z)]+p\E[W^+]-z$ for the example of $\delta=0.0839$.}\label{fig:influenceOfDelta:influenceOfP}
\end{figure}

Instead of the absence of contagious links, Theorem \ref{threshold:res} postulates 
certain threshold requirements to make our network model resilient to small initial shocks. Keeping the previously specified network parameters unchanged, we compute $\alpha_\text{c}\approx 2.13$ and $\gamma_\text{c}\approx0.468$. A natural choice for the threshold of bank $i\in[n]$ is then $\tau_i=\max\{2,\lfloor \alpha (w_i^-)^\gamma\rfloor\}$, where $\alpha=\alpha_\text{c}(1+\delta)$, $\gamma=\gamma_\text{c}(1+\delta)$ and $\delta\in[-1,\infty)$ denotes a (possibly negative) buffer. By Theorems \ref{threshold:res} and \ref{thm:functional:nonres}, networks are resilient to initial shocks for $\delta>0$ and non-resilient for $\delta<0$. The influence of $\delta$ on $f(z)$ can be seen in Figure \ref{fig:influenceOfDelta}. In particular, one notes that resilience for positive $\delta$ stems from the negative hump of $f(z)$ subsequent to zero. Further note, however, that resilience is only guaranteed to shocks whose size tends to zero. If the network is shocked by a strictly positive initial default fraction $p$, the final default cluster will only be small if $\delta>\delta_p$ for some $\delta_p>0$. This is because the functional $f(z)=(1-p)\E[W^+\psi_T(W^-z)]+p\E[W^+]-z$ for a uniformly shocked network depends on $p$. The influence of $p$ on $f(z)$ can be seen in Figure \ref{fig:influenceOfP}. In order for a network to be resilient to an initial shock of size $pn$ the hump subsequent to $0$ needs to become negative in Figure \ref{fig:influenceOfP}. It is always possible to determine numerically the least necessary buffer $\delta$ to make a system resilient to a shock of initial default fraction $p$ or less (see Table \ref{tab:deltaFromP} for the corresponding values of $\delta$ for $p=0.001k$, $k\in[10]$). Note that a buffer of $\delta=0.0839$ yields $\alpha=2.31$ and $\gamma=0.507$ and hence the thresholds required to make the system resilient to shocks of $1\%$ are still strongly sublinear.
\begin{table}[h]
	\caption{List of values for buffer $\delta$ corresponding to initial default of $\lfloor pn\rfloor$ banks}
		\hfill\begin{tabular}{l|c|c|c|c|c|c|c|c|c|c}
			$p$ $[\%]$ & 0.1 & 0.2 & 0.3 & 0.4 & 0.5 & 0.6 & 0.7 & 0.8 & 0.9 & 1.0\\\hline
			$\delta$ $[\%]$ & 2.35 & 3.44 & 4.30 & 5.04 & 5.71 & 6.36 & 6.89 & 7.42 & 7.91 & 			8.39
		\end{tabular}
	\label{tab:deltaFromP}
	\hfill
	\vspace{-0mm}
\end{table}

\begin{figure}[t]
    \hfill\subfigure[]{\includegraphics[width=0.4\textwidth]{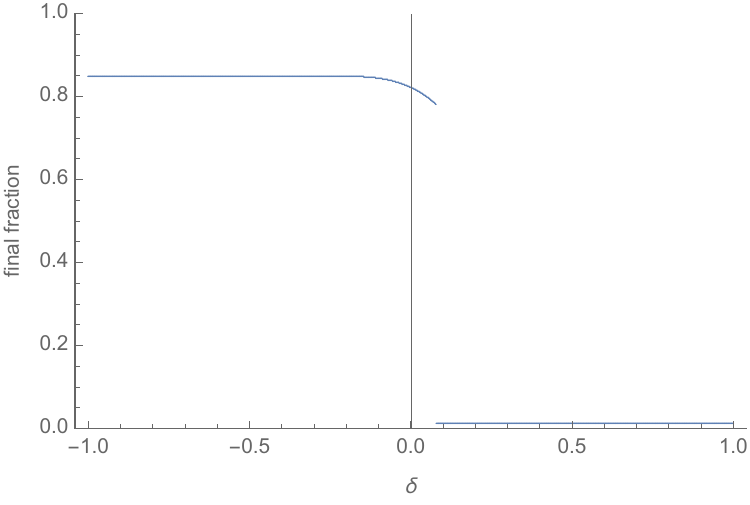}\label{fig:SimulationDelta1000000}}
    \hfill\subfigure[]{\includegraphics[width=0.4\textwidth]{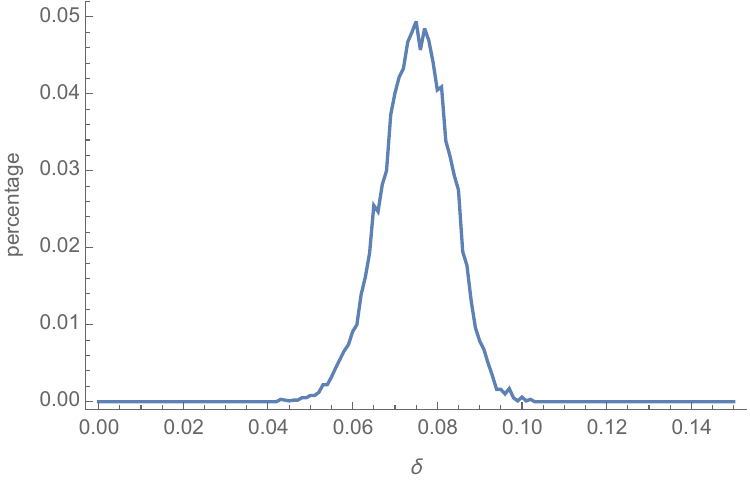}\label{fig:DeltaDistribution_0.01_1000000}}\hfill
\caption{(a) A typical result for the final fraction in a network of $10^6$ banks with initial default fraction of $p=1\%$ as $\delta$ varies between $-1$ and $1$ in steps of $10^{-3}$. (b) The distribution of jump points for $10^4$ networks of size $n=10^6$ with initial default fraction $p=1\%$.}\label{fig:SimulationDelta1000000:DeltaDistribution_0.01_1000000}
\end{figure}

\noindent We want to verify above results by simulations. For this, we simulate a very large network consisting of $n=10^6$ banks and keeping the network topology constant we let $\delta$ vary between $-1$ and $1$ in steps of $10^{-3}$. For each simulated network, we then find that it becomes resilient for $\delta$ large enough. This becomes visible by a jump of the final fraction of defaulted banks at this particular $\delta$ as illustrated in Figure \ref{fig:SimulationDelta1000000} for a sample network. 

Keeping track of the values of $\delta$ at which the final fraction drops near $p=1\%$ for $10^4$  simulated networks yields the distribution shown in Figure \ref{fig:DeltaDistribution_0.01_1000000}. It shows a peak at about \mbox{$\delta=0.076$} and hence supports our theoretical findings from above. Deviations from the theoretical value $\delta_{0.01}\approx 0.0839$ are small and can be explained by the finite (albeit very large) network size.

\begin{figure}[t]
	\hfill\includegraphics[width=0.4\textwidth]{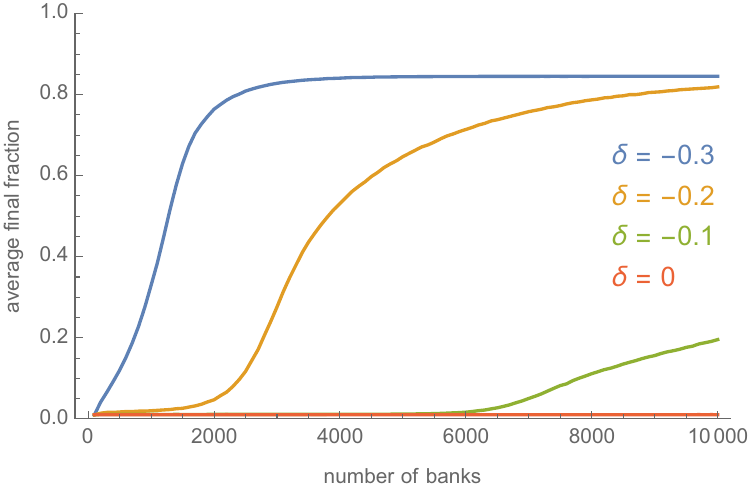}\hfill
		\caption{Average final fraction of defaulted banks in finite networks}
		\label{fig:gemittelt}
\end{figure}

Having looked at the theoretical capital requirements for very large networks, it is now sensible to turn our attention to networks of a few thousand banks as they arise in the real world. Figure \ref{fig:gemittelt} shows the final fraction in networks of size $n\leq10^4$ with initial default fraction $p=0.01$ for $\delta$ between $-0.3$ and $0$. For each $n$, we averaged over $10^5$ simulations. The figure shows that networks of size $n\leq10^4$ are already resilient for $\delta=0$. Even for $\delta=-0.2$ the network is rather resilient if $n\leq 2,000$ resp.~for $\delta=-0.1$ if $n\leq 6,000$. That is, our result is robust in the sense that already lower threshold requirements are sufficient to make the systems resilient to small shocks. The deviations stem from the relatively small network sizes of only a few thousand. Here, rare extreme values of vertex weights fail to appear despite the missing second moment condition or those large banks are not infected by the uniform initial infection. 

For managing systemic risk in real networks it might, however, be of interest not only how some uniform initial default influences the system but also how the default of the largest banks does. In a further simulation, we hence choose the $\lfloor pn\rfloor$ largest (by weights) banks in the network to default at the beginning. The function $f(z)$ then qualitatively keeps its shape as in Figure \ref{fig:influenceOfDelta:influenceOfP} but is shifted upwards. Again, we can compute corresponding values of $\delta$ and $p$ numerically. We list our results in Table \ref{tab:deltaFromPWithLargest}. As one expects, the values of $\delta$ are larger in this case than the ones we obtained for uniform infection in Table \ref{tab:deltaFromP}, but only by a factor of about $2$ and as before the resulting capital requirements are strongly sublinear.

\begin{table}[h]
	\caption{List of values for buffer $\delta$ corresponding to initial default of the $\lfloor pn\rfloor$ largest banks}
		\hfill\begin{tabular}{l|c|c|c|c|c|c|c|c|c|c}
		$p$ $[\%]$& 0.1 & 0.2 & 0.3 & 0.4 & 0.5 & 0.6 & 0.7 & 0.8 & 0.9 & 1.0\\\hline
		$\delta$ $[\%]$& 4.09 & 6.05 & 7.61 & 8.90 & 10.0 & 11.0 & 11.9 & 12.7 & 13.4 & 14.1
		\end{tabular}\hfill
	\label{tab:deltaFromPWithLargest}
\end{table}

\subsection{Simulations for the Exposure Model}

We can now turn to the simulation of a weighted network as in the exposure model. In addition to the network parameters of the threshold model in the previous subsection, we assume that for $i\neq j$, exposures $E_{j,i}$ are given by $E_{j,i}\stackrel{d}{=}E_i$ for Pareto distributed random variables $E_i$ with exponent $\xi=2.5277$ as in \cite{Cont2013} and minimal value $E_{\text{min},i}$. The exposures are assumed independent of each other and the network topology. The minimal exposures $E_{\text{min},i}$, can be chosen arbitrarily since they act as a constant factor for all exposures $E_{j,i}$ and capital $c_i$.

\begin{figure}[t]
    \hfill\subfigure[]{\includegraphics[width=0.4\textwidth]{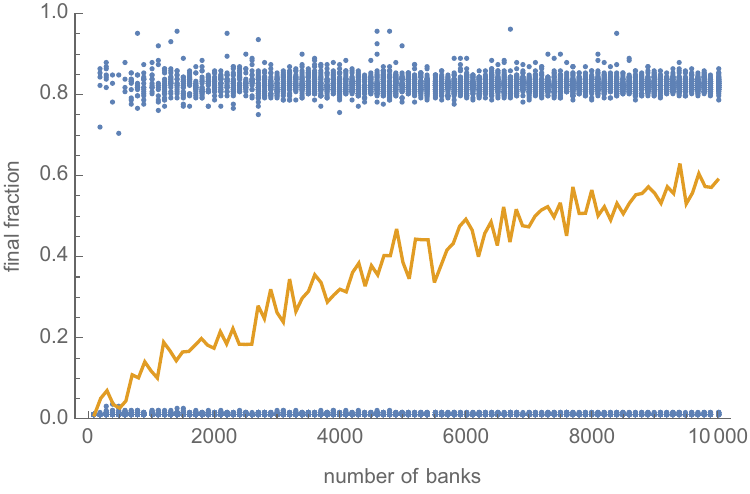}\label{fig:exposure:model}}
    \hfill\subfigure[]{\includegraphics[width=0.4\textwidth]{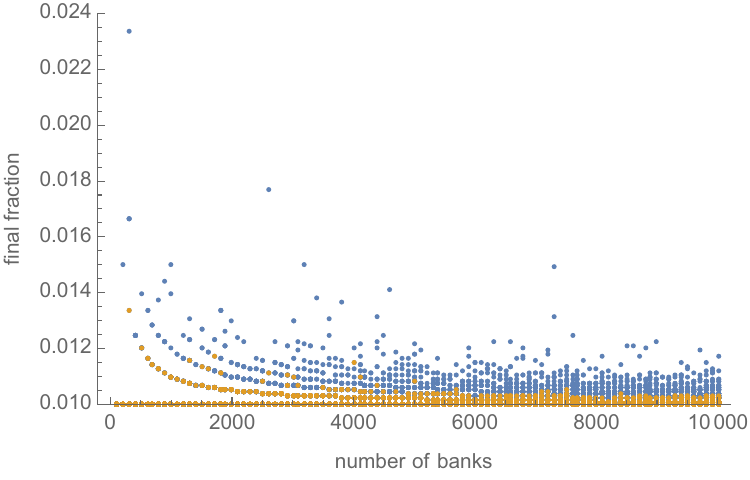}\label{fig:capital:requirements}}\hfill
\caption{(a) Scatter plot of the final fraction of defaulted banks for weighted networks of finite size without contagious links. Orange: average over all $100$ configurations for \mbox{each size.} \mbox{(b) Scatter} plot of the final fraction of defaulted banks for weighted networks of finite size. Blue: Capitals determined by Theorem \ref{cor:threshold:res}. Orange: Capitals determined by Proposition \ref{prop:robust:capital:requirements}}
		\label{fig:exposure:model:capital:requirements}
\end{figure}

In a first simulation, again we assume absence of contagious links but nothing more. That is, we first simulate the network skeleton and the edge-weights independently and then determine the banks' capitals as their largest exposure value plus some small buffer $\epsilon>0$. For our simulation, we choose $\epsilon=10^{-3}\E[E_i]=10^{-3}E_{\text{min},i}(\xi-1)/(\xi-2)$. As before, we assume initial default of $p=1\%$ uniformly chosen banks in the network and simulate the default process for $n\in\{100k\,:\, k\in[100]\}$ and $100$ different configurations of the random network for each $n$. The results for the final fraction of defaulted banks are plotted in Figure \ref{fig:exposure:model}. We notice that already for small network sizes there are some non-resilient network samples with final default fraction of about $80\%$. As the number of banks $n$ grows, also the probability that the networks are non-resilient significantly increases. This can be seen from the orange curve in Figure \ref{fig:exposure:model} which shows the average final default fraction taken over all $100$ configurations. The simulation supports our analytical result that for networks without a second moment condition on their degree sequences, only the absence of contagious links does not ensure resilience.

In a second simulation, we keep the network topology and the exposure sizes from the first simulation unchanged and choose capitals according to the formula in Proposition \ref{prop:robust:capital:requirements} with $\tau(w)=\max\{2,\lfloor\alpha w^\gamma\rfloor\}$ for $\alpha=\alpha_\text{c}(1+\delta)$, $\gamma=\gamma_\text{c}(1+\delta)$ and $\delta=8.39\%$ as in Table \ref{tab:deltaFromP}. As can be seen from Figure \ref{fig:capital:requirements}, already for typical network sizes of less than $10^4$, these capital allocations make the system resilient (note the axis scale). The maximal final fraction we observed was given by $1.33\%$. As mentioned before, the capital requirements in Proposition \ref{prop:robust:capital:requirements} are too conservative in general, however. In another simulation, we hence choose capitals as determined in Theorem \ref{cor:threshold:res} again for $\tau(w)=\max\{2,\lfloor\alpha w^\gamma\rfloor\}$. Figure \ref{fig:capital:requirements} shows that under these requirements the fundamental defaults still do not spread through the network. All observed final fractions were less or equal $2.33\%$. However, keeping track of the total capitalization of the system further reveals that the capital requirements from Theorem \ref{cor:threshold:res} only amount to about $61\%$ of the ones from Proposition \ref{prop:robust:capital:requirements} for our chosen network parameters.

\section{Proofs}\label{sec:proofs}
\subsection{Proofs for Section \ref{random:graph}}\label{ssec:proofs:2}
\begin{proof}[Proof of Lemma \ref{lem:f:continuous}]
Continuity of $f$ follows directly from Lebesgue's dominated convergence theorem, noting that $W^+$ is integrable by Assumption \ref{vertex:assump}. Further, $f(0)=\E[W^+\1_{\{T=0\}}]>0$ and $\lim_{z\to\infty}f(z;(W^-,W^+,T))=-\infty$. Hence by the intermediate value theorem function $f$ must have a positive root $\hat{z}$. Representation (\ref{eqn:integral:representation}) follows by an application of Fubini's theorem:
\begin{align}\label{eqn:integral:representation}
f(z) &= \E\left[W^+\1_{\{T=0\}} + \int_0^z W^-W^+\P\left(\mathrm{Poi}(W^-\xi)=T-1\right)\1_{\{T\geq1\}}\dd\xi\right]-z\nonumber\\
&= \E\left[W^+\1_{\{T=0\}}\right] + \int_0^z\left(\E\left[W^-W^+\phi_T(W^-\xi)\right]-1\right)\dd\xi\qedhere
\end{align}
\end{proof}

\begin{proof}[Proof of Theorem~\ref{thm:asymp:1}]
We want to make use of Theorem \ref{thm:threshold:model} for the threshold model. 
Thus we describe an alternative description of default contagion compared to Subsection \ref{ssec:default:contagion:systemic:importance}:

At the beginning we declare all initially defaulted vertices to be \emph{defaulted} but yet \emph{unexposed}. At each step, a single defaulted, unexposed vertex $i \in [n]$ is picked and exposed to its neighbors, i.\,e.~weighted edges to its neighbors are drawn. If bank $j$ goes bankrupt due to the new edge that is sent from $i$, it is added to the set of defaulted, unexposed vertices. Otherwise, the capital of $j$ is reduced by the amount $e_{i,j}$. Afterwards, we remove $i$ from the set of unexposed vertices.

We keep track of the following sets and quantities at different steps $0\leq t\leq n-1$:
\begin{enumerate}[leftmargin=*,label=\alph*.]
\item $U(t)\subset[n]$: the unexposed vertices at step $t$. We set $U(0):=\{ i \in [n] \,:\, c_i=0 \}$.
\item $N (t)\subset[n]$: the solvent vertices at step $t$. At $t=0$, we set $N(0):=[n]\backslash U(0)$.
\item The updated capitals $\{\tilde{c}_i (t) \}_{i\in [n]}$ with $\tilde{c}_i(0)=c_i$ for all $i\in[n]$.
\end{enumerate}
At step $t\in [n-1]$ the sets and quantities are updated according to the following scheme:
\begin{enumerate}[leftmargin=*]
\item \label{chose:rule} Choose a vertex $v\in U(t-1)$ according to any rule.
\item \label{update:capitals:rule} Expose $v$ to all of its neighbors in $N (t-1)$. That is, for all vertices $w\in N(t-1)$ set $\tilde{c}_w (t):=\max \{ 0 ,\tilde{c}_w (t-1) - e_{v,w} \}$. Note that $\tilde{c}_w (t)= \tilde{c}_w (t-1)$ if $e_{v,w}=0$.
\item \label{update:sets:rule} Set $N (t):= \{ i \in N(t-1) \,:\,  \tilde{c}_i (t)> 0 \}$ and $U(t):= ( U(t-1)\setminus \{v\} ) \cup \{ i \in N(t-1) \,:\,  \tilde{c}_i (t)= 0 \}$.
\end{enumerate}
Edges that are sent to already insolvent vertices are not exposed (but they could). Above steps are repeated until step $\hat{t}$, the first time that $U(t)=\emptyset$. Note that $\hat{t}$ is the final number of infected vertices independent of the rule chosen in Step~\ref{chose:rule}. Further, we can complete the exposition of the entire graph by exposing also links to defaulted vertices and links sent from vertices in $N (\hat{t})$.

Now, observe that the rule chosen in \ref{chose:rule} defines a permutation of the $\hat{t}$ elements of $[n]$ that go bankrupt. Further, for each $j \in [n]$ it defines an ordering of the set of insolvent vertices that send an edge to $j$, describing the order in which the edges are exposed. This ordering can be completed to a bijective map $\pi_j:[n-1]\to[n]\setminus \{ j \}$ by adding vertices that either send no edge to $j$ or are still solvent in the end. To be precise, let $\pi_j$ denote the ordering for vertex $j$ and let this vertex (after the exposition) have $l$ links sent from insolvent vertices. Then the entries $\pi_j(1),\dots , \pi_j(l)$ list defaulted neighbors in $[n]\backslash\{j\}$ in the order their edges are sent to vertex $j$. The entries $\pi_j(l+1),\dots,\pi_j(n-1)$ are, in their natural order, the remaining vertices in $[n]\backslash\{j\}$.

In order to reduce the model to the threshold model from Subsection \ref{ssec:special:case:threshold:model}, we now want to give a meaning to the so far only hypothetical threshold values $\tau_i$, $i\in[n]$. The idea is to construct a new random graph that has the same distribution (also of the threshold) as the graph constructed in Subsection~\ref{ssec:exposure:model} but with thresholds that have a direct meaningful interpretation:

We work on the same probability space as before but instead of assigning weight $E_{i,j}$ to a potential edge sent from $i\in[n]\backslash\{j\}$ to $j$, now the $i-$th ($i\in[n-1]$) edge that is sent to vertex $j$ during the sequential exposition described above shall receive weight $E_{\rho_j(i),j}$, where as before $\rho_j$ is the natural enumeration of $[n]\backslash\{j\}$. That is, edge-weights are not linked to the natural indices of their vertices anymore, but instead to the order of the exposition of the edges. One notes, however, that the random graph constructed that way has the same distribution as the random graph constructed before. To see this, observe that by the sequential procedure described by the orderings $\{\pi_j\}_{j\in[n]}$ and the assignment of exposures as described above, a potential edge sent from vertex $i$ to vertex $j$ is now assigned the edge-weight $E_{\rho_j(\pi_j^{-1}(i)),j}$. By exchangeability of the lists $\{E_{i,j}\}_{i\in[n]\backslash\{j\}}$ for $j\in[n]$, the new random variables $\{E_{\rho_j(\pi_j^{-1}(i)),j}\}_{i\in[n]\backslash\{j\}}$ have the same multivariate distribution as $\{E_{i,j}\}_{i\in[n]\backslash\{j\}}$. Obviously, also the new exposures are independent of the edge-indicator functions $\{X_{i,j}\}_{i,j\in[n]}$.

Hence, both constructions result in the same distribution for the random graph. Further, note that the assignment of edge-weights has been conducted in such a way that the threshold values in both versions of the network coincide. As before, they are given as
\[ \tau_i(n):= \inf \{ s \in \{0\}\cup[n-1] \,:\, \sum_{l\leq s} E_{\rho_i(l),i} \geq c_i \},\quad i\in[n].\]

\noindent In the new random graph, however, the thresholds $\tau_i$ have the interpretation of actual thresholds meaning that bank $i$ goes bankrupt at the $\tau_i$-th of one of its neighbors' default. The sequential description of the cascade process has then the advantage that we can reduce it to the threshold model as described in Subsection \ref{ssec:special:case:threshold:model}. We can replace the capitals $\tilde{c}_i (t)$, which represent monetary thresholds, by integer values $\tilde{\tau}_i(t)$ (we set $\tilde{\tau}_i(0):=\tau_i$), which count numbers of neighbors, and alter Steps \ref{update:capitals:rule} and \ref{update:sets:rule} in the description of the sequential cascade process according to the rule that if there is an edge sent from $v$ to $w$ ($e_{v,w}>0$), then set $\tilde{\tau}_w (t):=\tilde{\tau}_w (t-1) - 1$. If there is no edge from $v$ to $w$ ($e_{v,w}=0$), set $\tilde{\tau}_w (t):= \tilde{\tau}_w (t-1)  $.
Then the sets $N (t)$ and $U(t)$ are defined by $N (t):= \{ i \in N(t-1) \,:\,  \tilde{\tau}_i (t)> 0 \}$ respectively $U(t):= \left(U(t-1)\backslash\{v\}\right)\cup \{ i \in N(t-1) \,:\,  \tilde{\tau}_i (t)= 0 \}$. Everything else stays unchanged. Note that the resulting threshold values $\tilde{\tau}_i(t)$ are only valid for an exposition in the order as specified above. In the threshold model, however, we are free to choose exactly the same rule as we chose in Step \ref{chose:rule} of our model since this does not affect the final set of defaulted vertices. Hence, we can replace our exposure model by the threshold model from Subsection \ref{ssec:special:case:threshold:model}, resulting in the same final set of defaulted vertices.
In Theorem \ref{thm:threshold:model}, a regularity condition on the capital (here threshold) distribution $T$ is required. This is ensured after conditioning on the values of $\{\tau_i\}_{i\in[n]}$ by Assumption~\ref{vertex:assump} that $G_n(x,y,v,l)$ converges to $G(x,y,v,l)$ almost surely for all $(x,y,v,l)$. Applying Theorem \ref{thm:threshold:model} hence yields the desired statement.
\end{proof}

\begin{proof}[Proof of Theorem~\ref{thm:asymp:2}]
Part \ref{thm:asymp:2:1} follows from Theorem \ref{thm:threshold:model} by the same arguments as before.

In order to prove the second part, we will apply an additional small shock to the system such that each bank $i$, regardless of its attributes $w_i^-$, $w_i^+$ and $\tau_i$, has its capital $c_i$ and hence its threshold $\tau_i$ set to $0$ with probability $p$, where $p$ is some fixed small number. The new limiting distribution of the system is then given by $(W^-,W^+,TM_p)$, where $M_p$ is a $\{0,1\}$-valued random variable independent of $(W^-,W^+,T)$ and with $\P(M_p=0)=p$. Instead of $f(z)=f(z;(W^-,W^+,T))$ we then have to consider the function
\[ f_p(z):=f(z;(W^-,W^+,TM_p)) = p(\E[W^+]-z)+(1-p)f(z). \]
Assuming $P(T=0)<1$ (the case $\P(T=0)=1$ is trivial), it holds $f_p(z)>f(z)$ and hence we conclude that the first positive root $\hat{z}_p$ of $f_p(z)$ is larger than $z^*$. By definition of $z^*$ we further derive that $\hat{z}_p\to z^*$ as $p\to0$. The idea is therefore to choose $p$ in such a way that $\hat{z}_p$ satisfies the assumptions of Theorem \ref{thm:asymp:1} and then conclude by coupling the original system with the additionally shocked one to derive $n^{-1}\mathcal{S}_n\leq n^{-1}\mathcal{S}_n^{(p)}$, where $\mathcal{S}_n^{(p)}:=\sum_{i\in\mathcal{D}_n^{(p)}}s_i$ and $\mathcal{D}_n^{(p)}$ denotes the set of finally defaulted vertices in the additionally shocked system. Since $z^*$ is a root of the continuously differentiable function $f(z)$ it must hold $d(z^*)\leq0$. We distinguish two cases:

In the first case, we assume that $\kappa:=d(z^*)<0$. Then by continuity of $d(z)$ on a neighborhood of $z^*$, also
\[ d_p(z):=\E\left[W^-W^+\phi_{TM_p}(W^-z)\right]-1 
 < \frac{\kappa}{2} < 0 \]
on a neighborhood of $\hat{z}_p$ for $p$ small enough. As indicated above, an application of Theorem \ref{thm:asymp:1} together with a coupling argument then yields
\begin{align*}
n^{-1}\mathcal{S}_n &\leq n^{-1}\mathcal{S}_n^{(p)}= \E\left[S\psi_{TM_p}(W^-\hat{z}_p)\right] + o_p(1)\leq \E\left[S\psi_T(W^-\hat{z}_p)\right] + \E[S\1_{\{M_p=0\}}] + o_p(1)\\
&\leq \E\left[S\psi_T(W^-z^*)\right] + \epsilon + o_p(1)
\end{align*}
by continuity of $\E\left[S\psi_T(W^-z)\right]$ ($S$ is assumed integrable) and choosing $p$ small enough.

In the second case, we have $0 = d(z^*) = \lim_{\epsilon\to0}\epsilon^{-1} f(z^*+\epsilon)$. For $\tilde{\epsilon}>0$, let
\[ \delta(\tilde{\epsilon}) :=  -\inf_{0<\epsilon\leq\tilde{\epsilon}}\epsilon^{-1}f(z^*+\epsilon),\]
which is positive for all $\tilde{\epsilon}$ by definition of $z^*$. We can therefore find $\tilde{\epsilon}>0$ arbitrarily small such that $\delta(\epsilon)<\delta(\tilde{\epsilon})$ for all $\epsilon<\tilde{\epsilon}$. We then derive that
\[ 0 \leq f(z^*+\epsilon) + \delta(\epsilon)\epsilon \leq f(z^*+\epsilon) + \delta(\tilde{\epsilon})\epsilon \]
for all $\epsilon\leq\tilde{\epsilon}$ with equality only for $\epsilon=\tilde{\epsilon}$. Hence at $\epsilon=\tilde{\epsilon}$ the derivative of the last term must be non-positive, i.\,e.~$d(z^*+\tilde{\epsilon})\leq -\delta(\tilde{\epsilon}) < 0$. By continuity, also $d(z) \leq -\delta(\tilde{\epsilon})/2<0$ on a neighborhood of $z^*+\tilde{\epsilon}$. Hence $z^*+\tilde{\epsilon}$ is a good candidate for the first positive root of the additionally shocked system. All that is left to show is that there exists a certain value for the shock size $p$ such that $z^*+\tilde{\epsilon}$ becomes the first positive root. To this end, let
\[ p(\tilde{\epsilon}) := \frac{\tilde{\epsilon}\delta(\tilde{\epsilon})}{\E[W^+]-z^*-\tilde{\epsilon}(1-\delta(\tilde{\epsilon}))}. \]

\noindent Note that for $\P(T=0)<1$ the root $z^*$ is always less than $\E[W^+]$ and hence for $\tilde{\epsilon}$ small enough $p(\tilde{\epsilon})$ becomes positive. As $\tilde{\epsilon}\to0$, also $p(\tilde{\epsilon})$ tends to zero. Now note that for all $0<\epsilon\leq\tilde{\epsilon}$,
\[ f_{p(\tilde{\epsilon})}(z^*+\epsilon) \geq (1-p(\tilde{\epsilon}))(-\epsilon\delta(\tilde{\epsilon})) + p(\tilde{\epsilon})\left(\E[W^+]-(z^*+\epsilon)\right) \geq (\tilde{\epsilon}-\epsilon)\delta(\tilde{\epsilon}) \geq 0 \]
with equality only for $\epsilon=\tilde{\epsilon}$. The additional shock strictly increases $f(z)$ and hence there cannot be any root less or equal $z^*$. In particular, $z^*+\tilde{\epsilon}$ is the first positive root of the additionally shocked system. By letting $\tilde{\epsilon}\to0$, we conclude for arbitrarily small $\epsilon>0$ that
\[ n^{-1}\mathcal{S}_n \leq \E\left[S\psi_T(W^-z^*)\right] + \epsilon + o_p(1). \]

\noindent For the case of $\hat{z}=z^*$, we simply need to combine parts \ref{thm:asymp:2:1} and \ref{thm:asymp:2:2} of the theorem.
\end{proof}

\subsection{Proofs for Section \ref{sec:resilience}}\label{ssec:proofs:3}
\begin{proof}[Proof of Proposition \ref{prop:convergence:speed}]
By (\ref{eqn:integral:representation}), we derive $f_i(z):=f(z;(W^-,W^+,TM_i))\leq \E[W^+\1_{\{M_i=0\}}] + \kappa z + o(z)$ and similarly $\E[S\psi_{TM_i}(W^-z)] \leq \E[S\1_{\{M_i=0\}}] + \kappa_S z + o(z)$.
Hence we derive
\[ \hat{z}_i \leq -\kappa^{-1}\E[W^+\1_{\{M_i=0\}}] + o(\E[W^+\1_{\{M_i=0\}}]), \]
where $\hat{z}_i$ denotes the first positive root of $f_i(z)$. Together with Theorem \ref{thm:asymp:1} this shows that
\[ \text{w.\,h.\,p.} \quad n^{-1}\mathcal{S}_n^{M_i} \leq \E[S\1_{\{M_i=0\}}] - \kappa^{-1}\kappa_S\E[W^+\1_{\{M_i=0\}}] + o(\E[W^+\1_{\{M_i=0\}}]). \]

\noindent If $f(z)$ and $\E[S\psi_T(W^-z)]$ are continuously differentiable from the right at $z=0$ with derivatives $\kappa<0$ and $\kappa_S<\infty$, then above inequalities are equalities and hence
\[ n^{-1}\mathcal{S}_n^{M_i} \xrightarrow{p} \E[S\1_{\{M_i=0\}}] - \kappa^{-1}\kappa_S\E[W^+\1_{\{M_i=0\}}] + o(\E[W^+\1_{\{M_i=0\}}]).\qedhere \]
\end{proof}

\begin{proof}[Proof of Theorem \ref{thm:cont:res}]
By condition 
(\ref{thm:cont:res:ass}), we derive that $z^*_i\to0$ as $i\to\infty$, where
\[ z^*_i := \inf\left\{ z>0\,:\,f(z;(W^-,W^+,TM_i))<0\right\}. \]
For $i$ large enough such that $z^*_i<z_0$, we can then apply part \ref{thm:asymp:2:2} of Theorem \ref{thm:asymp:2} to derive that
\[ \text{w.\,h.\,p.} \quad n^{-1}\mathcal{S}_n^{M_i} \leq \E\left[S\psi_{TM_i}(W^-z^*_i)\right] + \frac{\epsilon}{2} \leq \E\left[S\psi_T(W^-z^*_i)\right] + \E\left[S\1_{\{M_i=0\}}\right] + \frac{\epsilon}{2}.\]
Note that from continuity of $d(z)$ it follows that also $\E\left[W^-W^+\phi_{TM_i}(W^-z)\right]$ is continuous by dominated convergence. Since $S$ is integrable, the first summand tends to zero as $z^*_i\to0$ and also the second summand vanishes as $\P(M_i=0)\to0$. In particular, we can choose $i$ large enough such that
\[ \E\left[S\psi_T(W^-z^*_i)\right] + \E\left[S\1_{\{M_i=0\}}\right]\leq\frac{\epsilon}{2}.\qedhere \]
\end{proof}

We now turn to the proofs of Theorems \ref{threshold:res} -- \ref{thm:Pareto:type}. To show resilience of the financial system in Theorem \ref{threshold:res}, we want to use Theorem \ref{thm:cont:res} and thus need to ensure that $d(z)$ is continuous for $z>0$. This is done in the following lemma.

\begin{lemma}\label{lem:cont:differentiability}
Assume that $T=\tau(W^-)$ for some integer-valued function $\tau$ satisfying $\tau(w)=o(w)$. Then $d(z)$ is continuous on $(0,\infty)$.
\end{lemma}
\begin{proof}
We fix some $\tilde{z}<\infty$ and aim to show that for some arbitrarily fixed $\delta<\tilde{z}$ the family $\left\{W^-W^+\phi_T(W^-z)\right\}_{z\in[\tilde{z}-\delta,\tilde{z}+\delta]}$ is bounded by some integrable random variable almost surely. This will show continuity of $d(z)$ by Lebesgue's dominated convergence theorem.

By definition of a Poisson random variable, we derive
\[ W^-W^+\phi_T(W^-z) = W^-W^+\P\left(\mathrm{Poi}(W^-z)=T-1\right) = W^-W^+e^{-W^-z}\frac{(W^-z)^{T-1}}{(T-1)!}. \]
Then, by applying the well-known bound $(T-1)! \ge ((T-1)/e)^{T-1}$ we obtain
\[ W^-W^+\phi_T(W^-z) \leq W^-W^+\exp\left\{-W^-z\left(1-\frac{T-1}{W^-z}+\frac{T-1}{W^-z}\log\left(\frac{T-1}{W^-z}\right)\right)\right\}. \]
In the exponent, we identify the expression $g\left((T-1)/(W^-z)\right)$, where $g(x):=1-x+x\log(x)$. The continuous function $g$ admits the unique minimum $g(x^*)=0$ at $x^*=1$. Since further $\lim_{x\to0+}g(x)=1$, it holds that $g(x)\geq G$ for $x<1/2$ and some suitable $G>0$. Now choose $\tilde{w}$ large enough such that $(\tau(w)-1)/(w(\tilde{z}-\delta))<1/2$ for all $w>\tilde{w}$. We derive,
\[ W^-W^+\phi_T(W^-z) \leq W^+\left(W^-\exp\left\{-W^-(\tilde{z}-\delta)G\right\}+\tilde{w}\right) \leq W^+M(\tilde{z}) \in L^1, \]
almost surely, where $M(\tilde{z})$ is a positive constant depending on $\tilde{z}$ only.
\end{proof}

\begin{remark}\label{rem:cont:continuous:differentiability}
In Theorem \ref{threshold:res}, we have $\tau(w)=\mathcal{O}(w^\gamma)$ for $0\leq\gamma<1$ and hence $\tau(w)=o(w)$.
\end{remark}

\begin{proof}[Proof of Theorem \ref{threshold:res}]
In order to ease notation, we will assume throughout all the proofs that $w_\text{min}^-=w_\text{min}^+=1$. The arguments for general $w_\text{min}^-$ and $w_\text{min}^+$ are completely analogue.

Since $\E[W^-W^+]$ is maximized for comonotone weights, i.\,e.~$W^+=F_{W^+}^{-1}(F_{W^-}(W^-))=(W^-)^{(\beta^--1)(\beta^+-1)}$, we get $$\E\left[W^-W^+\right] \leq \E\left[(W^-)^{1+\frac{\beta^--1}{\beta^+-1}}\right] = (\beta^--1)\int_1^\infty w^{\gamma_\text{c}-1}\dd w < \infty$$ for $\gamma_\text{c}<0$. By dominated convergence we conclude that $f$ is continuously differentiable on $[0,\infty)$ with $f'(z)=d(z)$. By $T\geq2$, in particular, $f'(0)=-1$ and hence by Theorem \ref{prop:res} the system is resilient to small shocks.

Now let $\gamma_\text{c}=0$ and $\alpha:=\liminf_{w\to\infty}\tau(w)>\alpha_\text{c}+1$. We assume that $\alpha<\infty$, otherwise we could truncate $\tau(w)$ at some $\N\ni\alpha>\alpha_\text{c}+1$. Since $\tau(w)\geq2$ is an integer-valued function, we observe that $\alpha\in\N\backslash\{0,1\}$ and $\tau(w)\geq\alpha$ for all $w>\tilde{w}$ and some constant $\tilde{w}>0$. Since $\psi_r(x)$ is monotonically decreasing in $r$, we derive that, as $z\to0$,
\begin{align*}
\E\left[W^+\psi_T(W^-z)\right] &\leq \E\left[W^+\psi_\alpha(W^-z) + W^+\psi_2(W^-z)\1_{\{W^-\leq\tilde{w}\}}\right]\leq 
\E\left[W^+\psi_\alpha(W^-z)\right] + o(z).
\end{align*}
Note that since $\psi_\alpha(x)$ is a strictly increasing function in $x$, this expression becomes maximized for comonotone dependence of $W^-$ and $W^+$. Since $\gamma_c = 0$, we derive
\begin{align*}
\limsup_{z\to0+}z^{-1}\E\left[W^+\psi_T(W^-z)\right] &\leq \limsup_{z\to0+}z^{-1} \sum_{k\geq\alpha} \E\left[(W^-)^\frac{\beta^--1}{\beta^+-1}e^{-W^-z}\frac{(W^-z)^k}{k!}\right]\\
&= \limsup_{z\to0+}\alpha_\text{c}\sum_{k\geq\alpha} \int_z^\infty \frac{x^{\frac{\beta^--1}{\beta^+-1}-\beta^-+k}e^{-x}}{k!}\dd x\\
&= \alpha_\text{c}\sum_{k\geq\alpha} \frac{\Gamma(k-1)}{k!} = \frac{\alpha_\text{c}}{\alpha-1}<1.
\end{align*}
In particular, $\limsup_{z\to0+}z^{-1}f(z)<0$. By Lemma \ref{lem:cont:differentiability} and Remark \ref{rem:cont:continuous:differentiability}, we also know that $d(z)$ is continuous for $z>0$ (we can simply cut off $\tau(w)$ at $\alpha$) and hence by Theorem \ref{thm:cont:res} we can conclude that the system must be resilient.

Finally, assume that $\gamma_\text{c}>0$ and $\alpha:=\liminf_{w\to\infty}w^{-\gamma_\text{c}}\tau(w) > \alpha_\text{c}$. We can then choose some $\alpha_\text{c}<\tilde{\alpha}<\alpha$ and $\tilde{w}<\infty$ such that $\tau(w)\geq\left\lceil\tilde{\alpha}w^{\gamma_\text{c}}\right\rceil$ for all $w>\tilde{w}$. Hence we derive that
\[ \E\left[W^+\psi_T(W^-z)\right] \leq \E\left[W^+\psi_{\left\lceil\tilde{\alpha}(W^-)^{\gamma_\text{c}}\right\rceil}(W^-z)\1_{\{W^->\tilde{w}\}}\right] + o(z). \]

\noindent By a Chernoff bound we get that $\psi_r(x)\leq(xe/r)^re^{-x}$ for $x<r$. Thus for $w\leq\left(\tilde{\alpha}^{-1}(1+\epsilon)z\right)^\frac{1}{\gamma_\text{c}-1}$ and $\epsilon>0$,
\[ \psi_{\left\lceil\tilde{\alpha}w^{\gamma_\text{c}}\right\rceil}(wz) \leq \left(\frac{wze}{\tilde{\alpha}w^{\gamma_\text{c}}}\right)^{\tilde{\alpha}w^{\gamma_\text{c}}}e^{-wz} = \exp\left\{-z^\frac{\gamma_\text{c}}{\gamma_\text{c}-1}g\left(wz^\frac{1}{1-\gamma_\text{c}}\right)\right\}, \]
where $g(x) := x-\tilde{\alpha}x^{\gamma_\text{c}}\log(ex^{1-\gamma_\text{c}}/\tilde{\alpha})$. For arbitrary $\lambda>0$, we can hence choose $\tilde{w}$ large enough such that for all $\tilde{w}<w\leq\left(\tilde{\alpha}^{-1}(1+\epsilon)ez\right)^\frac{1}{\gamma_\text{c}-1}$ it holds
\[ \psi_{\left\lceil\tilde{\alpha}w^{\gamma_\text{c}}\right\rceil}(wz) \leq \left(\frac{wze}{\tilde{\alpha}w^{\gamma_\text{c}}}\right)^{\tilde{\alpha}w^{\gamma_\text{c}}} \leq \left(\frac{wze}{\tilde{\alpha}w^{\gamma_\text{c}}}\right)^{1+\lambda}(1+\epsilon)^{-\tilde{\alpha}w^{\gamma_\text{c}}+1+\lambda} \leq \left(\tilde{\alpha}^{-1}ze\right)^{1+\lambda} \]
and
\[ \E\left[W^+\psi_{\left\lceil\tilde{\alpha}(W^-)^{\gamma_\text{c}}\right\rceil}(W^-z)\1_{\{\tilde{w}<W^-\leq \left(\tilde{\alpha}^{-1}(1+\epsilon)ez\right)^\frac{1}{\gamma_\text{c}-1}\}}\right] \leq \left(\tilde{\alpha}^{-1}ze\right)^{1+\lambda}\E[W^+] = o(z). \]
For $(\tilde{\alpha}^{-1}(1+\epsilon)ez)^\frac{1}{\gamma_\text{c}-1} < w \leq (\tilde{\alpha}^{-1}(1+\epsilon)z)^\frac{1}{\gamma_\text{c}-1}$, we have $g(wz^\frac{1}{1-\gamma_\text{c}}) \geq \delta$ for some $\delta>0$. Thus
\[ \E\left[W^+\psi_{\left\lceil\tilde{\alpha}(W^-)^{\gamma_\text{c}}\right\rceil}(W^-z)\1_{\{\left(\tilde{\alpha}^{-1}(1+\epsilon)ez\right)^\frac{1}{\gamma_\text{c}-1}<W^-\leq\left(\tilde{\alpha}^{-1}(1+\epsilon)z\right)^\frac{1}{\gamma_\text{c}-1}\}}\right]\leq \E[W^+]e^{-\delta z^\frac{\gamma_\text{c}}{\gamma_\text{c}-1}} = o(z). \]
Hence, as $z\to0$, only $W^->\left(\tilde{\alpha}^{-1}(1+\epsilon)z\right)^\frac{1}{\gamma_\text{c}-1}$ contributes to $z^{-1}\E[W^+\psi_T(W^-z)]$. On this set, by bounding with the comonotone dependence, we compute
\begin{align}\label{eqn:expectation:upper:part}
\E\Bigg[W^+\psi_T(W^-z)\1_{\{W^->\left(\tilde{\alpha}^{-1}(1+\epsilon)z\right)^\frac{1}{\gamma_\text{c}-1}\}}\Bigg] &\leq \E\Bigg[W^+\1_{\{W^->\left(\tilde{\alpha}^{-1}(1+\epsilon)z\right)^\frac{1}{\gamma_\text{c}-1}\}}\Bigg]\\
&\leq (\beta^--1)\int_{\left(\tilde{\alpha}^{-1}(1+\epsilon)z\right)^\frac{1}{\gamma_\text{c}-1}}^\infty w^{\frac{\beta^--1}{\beta^+-1}-\beta^-}\dd w= \alpha_\text{c}\frac{(1+\epsilon)z}{\tilde{\alpha}}\nonumber
\end{align}
Hence, $\limsup_{z\to0+}z^{-1}f(z) = \limsup_{z\to0+}z^{-1}\E[W^+\psi_T(W^-z)]-1<0$ by choosing $\epsilon>0$ small enough such that $\tilde{\alpha}^{-1}\alpha_\text{c}(1+\epsilon)<1$. This shows resilience by the same arguments as in part 2, noting that we can cut $\tau(w)$ at $w^\eta$ for some $\gamma_\text{c}<\eta<1$.
\end{proof}

\begin{proof}[Proof of Theorem \ref{thm:functional:nonres}]
Again, we simplify notation be setting $w_\text{min}^-=w_\text{min}^+=1$.
We start by proving the second statement. Let $\alpha:=\limsup_{w\to\infty}w^{-\gamma_\text{c}}\tau(w)$ and choose $\alpha < \tilde{\alpha} < \int_0^\infty\Lambda(x^{1-\beta^+})\dd x$ and $\tilde{w}<\infty$ such that $\tau(w)\leq\left\lfloor\tilde{\alpha}w^{\gamma_\text{c}}\right\rfloor$ for all $w>\tilde{w}$. Moreover, choose $\epsilon>0$ and $\delta>0$ such that $\tilde{\alpha}<(1-\epsilon)(1-\delta)\int_0^\infty\Lambda(x^{1-\beta^+})\dd x$ and let $z>0$ small enough such that $\tilde{w}\leq\left(\tilde{\alpha}^{-1}(1-\epsilon)z\right)^{{1}/(\gamma_\text{c}-1)}$ as well as $z<(\epsilon^2\delta)^{{1}/(\gamma_\text{c}-1)}(\tilde{\alpha}/(1-\epsilon))^{{1}/\gamma_\text{c}}$. For such $z$ and $w>(\tilde{\alpha}^{-1}(1-\epsilon)z)^{{1}/(\gamma_\text{c}-1)}$ we obtain
\[ \P\left(\mathrm{Poi}(wz)\geq\left\lfloor\tilde{\alpha}w^{\gamma_\text{c}}\right\rfloor\right) \geq 1-\P\left(\left\vert\mathrm{Poi}(wz)-wz\right\vert \geq \epsilon wz\right) \geq 1-\frac{1}{\epsilon^2wz} \geq 1-\frac{1-\epsilon}{\epsilon^2\tilde{\alpha}w^{\gamma_\text{c}}} \geq 1-\delta, \]
by Chebyshev's inequality. Therefore,
\begin{align}\label{eqn:int:cond:prob}
&\E\left[W^+\psi_T(W^-z)\right] \geq (1-\delta)\E\left[W^+\1_{\{W^->\left(\tilde{\alpha}^{-1}(1-\epsilon)z\right)^\frac{1}{\gamma_\text{c}-1}\}}\right]\nonumber\\
&\hspace{0.382cm}= (1-\delta)\int_0^\infty \P\left(W^+>x,W^->\left(\tilde{\alpha}^{-1}(1-\epsilon)z\right)^\frac{1}{\gamma_\text{c}-1}\right)\dd x\nonumber\\
&\hspace{0.382cm}= (1-\delta)\left(\tilde{\alpha}^{-1}(1-\epsilon)z\right)^{\frac{1}{\gamma_\text{c}-1}\frac{\beta^--1}{\beta^+-1}}\nonumber\\
&\hspace{3.675cm} \times\int_0^\infty \P\left(F_{W^+}(W^+)>1-x^{1-\beta^+}p(z),F_{W^-}(W^-)>1-p(z)\right)\dd x\nonumber\\
&\hspace{0.382cm}= \tilde{\alpha}^{-1}(1-\epsilon)(1-\delta)z \int_0^\infty \P\left(F_{W^+}(W^+)>1-x^{1-\beta^+}p(z) \,\middle\vert\, F_{W^-}(W^-)>1-p(z)\right)\dd x,
\end{align}
where we substituted $p(z) := (\tilde{\alpha}^{-1}(1-\epsilon)z)^{(1-\beta^-)/(\gamma_\text{c}-1)}$. Note that the conditional probability is bounded by $1\wedge x^{1-\beta^+}$. Hence, by Lebesgue's dominated convergence theorem
\[ z^{-1}\E\left[W^+\psi_T(w^-z)\right] \geq \tilde{\alpha}^{-1}(1-\epsilon)(1-\delta)\int_0^\infty\Lambda\left(x^{1-\beta^+}\right)\dd x + o(1) > 1 + o(1) \]
and thus $\E\left[W^+\psi_T(W^-z)\right]>z$ for $z$ small enough which implies non-resilience by Theorem \ref{prop:nonres}.

For resilience in part 2 we follow the same approach as in the proof of Theorem \ref{threshold:res} until we arrive at (\ref{eqn:expectation:upper:part}), which we can now evaluate with the same means as above. That is,
\[ z^{-1}\E\bigg[W^+\1\bigg\{W^->\left(\tilde{\alpha}^{-1}(1+\epsilon)z\right)^\frac{1}{\gamma_\text{c}-1}\bigg\}\bigg] \to \tilde{\alpha}^{-1}(1+\epsilon)\int_0^\infty\Lambda\left(x^{1-\beta^+}\right)\dd x <1,\quad\text{as }z\to0. \]
For the first statement of the theorem, note that we can lower bound the integral in (\ref{eqn:int:cond:prob}) by $\P(F_{W^+}(W^+)>1-p(z) \mid F_{W^-}(W^-)>1-p(z))$ and hence we derive non-resilience as above by
\[ \liminf_{z\to0}z^{-1}\E[W^+\psi_T(W^-z)] \geq \tilde{\alpha}^{-1}(1-\epsilon)(1-\delta) \lambda > 1.\qedhere \]
\end{proof}

\begin{proof}[Proof of Theorem \ref{thm:Pareto:type}]
Let $U\sim\mathrm{Unif}
[0,1]$ be uniformly distributed on the interval $[0,1]$ and
\[ X^\pm := (1-U)^\frac{1}{1-\beta^\pm}K^\pm. \]
Then $X^\pm\sim\mathrm{Par}(\beta^\pm,K^\pm)$. Further, let $\tilde{W}^\pm := F_{W^\pm}^{-1}(U)$ such that $\tilde{W}^\pm\stackrel{d}{=} W^\pm$ but $\tilde{W}^-$ and $\tilde{W}^+$ are comonotone. By \eqref{eqn:Pareto:type} we then now that $\tilde{W}^\pm=F_{W^\pm}^{-1}(U)\leq F_{X^\pm}^{-1}(U)=X^\pm$ for $U\geq\tilde{u}$ large enough and hence
\[ \E[W^-W^+]\leq\E[\tilde{W}^-\tilde{W}^+]\leq \tilde{w}(\E[W^-]+\E[W^+]) + \E[X^-X^+], \]
where $\tilde{w}=\max\{F_{W^-}^{-1}(\tilde{u}),F_{W^+}^{-1}(\tilde{u})\}<\infty$. Then for the case of $\gamma_c<0$, by the same calculations as in the proof of Theorem \ref{threshold:res} we derive $\E[X^-X^+]<\infty$ and thus resilience of the system.

For $\gamma_c=0$, as in the proof of Theorem \ref{threshold:res} we have
\[ \E[W^+\psi_T(W^-z)] \leq \E[W^+\psi_\alpha(W^-z)]+o(z), \]
where $\alpha=\liminf_{w\to\infty}\tau(w)\in\N\backslash\{0,1\}$. Hence
\begin{align*}
\limsup_{z\to0+} \frac{\E[W^+\psi_T(W^-z)]}{z} &\leq \limsup_{z\to0+} \frac{\E[\tilde{W}^+\psi_\alpha(\tilde{W}^-z)]}{z}\\
&\leq \limsup_{z\to0+} \frac{\E[X^+\psi_\alpha(X^-z)] + \tilde{w}\psi_2(\tilde{w}z)}{z}\\
&=\frac{\beta^+-1}{\beta^+-2}K^+K^-\frac{1}{\alpha-1}
\end{align*}
and the system is resilient for $\liminf_{w\to\infty}\tau(w)=\alpha>\frac{\beta^+-1}{\beta^+-2}K^+K^-+1$.

For $\gamma_c>0$, we follow the proof of Theorem \ref{threshold:res} and replace the right-hand side of \eqref{eqn:expectation:upper:part} by
\[ \E\bigg[X^+\1_{\{X^->\left(\tilde{\alpha}^{-1}(1+\epsilon)z\right)^\frac{1}{\gamma_\text{c}-1}\}}\bigg] + \tilde{w}\1_{\tilde{w}>\left(\tilde{\alpha}^{-1}(1+\epsilon)z\right)^\frac{1}{\gamma_\text{c}-1}} = \frac{\beta^+-1}{\beta^+-2}K^+(K^-)^{1-\gamma_c}\frac{(1+\epsilon)z}{\tilde{\alpha}} + o(z).\qedhere \]
\end{proof}

\begin{proof}[Proof of Proposition \ref{prop:robust:capital:requirements}]
The capitals $c_i$ are chosen such that the threshold values $\tau_i$ are at least $\tau(w_i^-)$. Hence coupling the weighted network to the corresponding threshold network yields the result.
\end{proof}

\begin{proof}[Proof of Theorem \ref{cor:threshold:res}]
By the assumption of $c_i>\max_{j\in[n]\backslash\{i\}}E_{j,i}$ almost surely for all $i\in[n]$, we get $\tau_i\geq2$ almost surely. The proof of part \ref{cor:threshold:res:1} is thus completely analogue to the one of Theorem \ref{threshold:res}.

We continue by proving part \ref{cor:threshold:res:3}. By the means of the proof of Theorem \ref{threshold:res} we derive that $\limsup_{z\to0+}z^{-1}\E\left[W^+\psi_T(W^-z)\1\left\{T>(1+\epsilon)\alpha_\text{c}(W^-)^{\gamma_\text{c}}\right\}\right] < 1$ for each $\epsilon>0$. It will hence suffice to prove $\E\left[W^+\psi_T(W^-z)\1\left\{T\leq(1+\epsilon)\alpha_\text{c}(W^-)^{\gamma_\text{c}}\right\}\right] = o(z)$ in order to show resilience. To this end, choose $0<\delta<\gamma_\text{c}(t-1)$. Then
\begin{align*}
\E\Big[W^+\psi_T(W^-z)\1\Big\{T\leq(W^-)^\delta\Big\}\Big] &\leq \E\left[W^+\psi_2(W^-z)\1\left\{T\leq(W^-)^\delta\right\}\right]\\
&= \lim_{M\to\infty}\E\left[\left(W^+\wedge M\right)\psi_2(W^-z)\1\left\{T\leq(W^-)^\delta\right\}\right]\\
&\leq \liminf_{M\to\infty} \liminf_{n\to\infty} n^{-1}\sum_{i\in[n]}\left(w_i^+\wedge M\right)\psi_2(w_i^-z)\1\left\{\tau_i\leq(w_i^-)^\delta\right\},
\end{align*}
where we approximated the non-continuous integrand by continuous ones and used almost sure weak convergence by Assumption \ref{vertex:assump}. Now taking expectations with respect to the exposure lists, by Fatou's lemma we derive
\begin{align*}
\E\left[W^+\psi_T(W^-z)\1\left\{T\leq(W^-)^\delta\right\}\right] &\leq \liminf_{M\to\infty} \liminf_{n\to\infty} n^{-1}\sum_{i\in[n]}\left(w_i^+\wedge M\right)\psi_2(w_i^-z)\P\left(\tau_i\leq(w_i^-)^\delta\right)\\
&\leq K_1 \liminf_{M\to\infty} \liminf_{n\to\infty} n^{-1}\sum_{i\in[n]}\left(w_i^+\wedge M\right)\psi_2(w_i^-z)(w_i^-)^{\delta-t\gamma_\text{c}}\\
&= K_1 \liminf_{M\to\infty} \E\left[\left(W^+\wedge M\right)\psi_2(W^-z)(W^-)^{\delta-t\gamma_\text{c}}\right]\\
&= K_1 \E\left[W^+\psi_2(W^-z)(W^-)^{\delta-t\gamma_\text{c}}\right],
\end{align*}
where we used Assumption \ref{ass:exposures} to bound
\begin{align*}
\P\left(\tau_i\leq(w_i^-)^\delta\right) &= \P\Bigg(\sum_{1\leq j\leq(w_i^-)^\delta}E_{\rho_i(j),i}\geq c_i\Bigg) \leq \P\Bigg(\sum_{1\leq j\leq(w_i^-)^\delta}E_{\rho_i(j),i}\geq \tau(w_i^-)\mu_i\Bigg)\\
& \leq \P\Bigg(\sum_{1\leq j\leq(w_i^-)^\delta}E_{\rho_i(j),i}\geq \alpha_\text{c}(w_i^-)^{\gamma_\text{c}}\mu_i\Bigg)\leq K_1(w_i^-)^{\delta-t\gamma_\text{c}},
\end{align*}
for $w_i^-$ large enough and some uniform constant $K_1>\infty$. Note that for $W^-\leq\tilde{w}$, we have $\E\left[W^+\psi_T(W^-z)\1_{\{W^-\leq\tilde{w}\}}\right] = o(z)$ as in the proof of Theorem \ref{threshold:res}. Hence it holds
\[ z^{-1}\E\left[W^+\psi_T(W^-z)\1\left\{T\leq(W^-)^\delta\right\}\right] \leq K_1 \E\left[W^+\frac{\psi_2(W^-z)}{W^-z}(W^-)^{1+\delta-t\gamma_\text{c}} \right] + o(1). \]
Since $\psi_2(x)=o(x)$, by dominated convergence it is enough to prove $\E\left[W^+(W^-)^{1+\delta-t\gamma_\text{c}} \right]<\infty$ in order for $\E\left[W^+\psi_T(W^-z)\1\left\{T\leq(W^-)^\delta\right\}\right]=o(z)$. We can easily choose $t>1$ and $\delta>0$ in such a way that $1+\delta-t\gamma_\text{c}>0$ and can therefore estimate $\E[W^+(W^-)^{1+\delta-t\gamma_\text{c}} ]$ by the comonotone expectation $\E\Big[(W^-)^{(\beta^--1)/(\beta^+-1)+1+\delta-t\gamma_\text{c}} \Big]$ which is finite since by our choice of $\delta$,
\[ \frac{\beta^--1}{\beta^+-1}+1+\delta-t\gamma_\text{c}-\beta^- = \gamma_\text{c}(1-t)+\delta-1 < -1. \]

\noindent Now let $2\leq N<\gamma_\text{c}/\delta$ and consider
$\E\left[W^+\psi_T(W^-z)\1\big\{(W^-)^{(N-1)\delta}<T\leq(W^-)^{N\delta}\big\}\right]$. As in the proof of Theorem \ref{threshold:res} it is enough to consider
$\E\left[W^+\1\big\{W^->z^{{1}/({(N-1)\delta-1})},T\leq(W^-)^{N\delta}\big\}\right]$.
Similar as above, we derive $\P\left(\tau_i\leq(w_i^-)^{N\delta}\right) \leq K_N(w_i^-)^{N\delta-t\gamma_\text{c}}$ for some uniform $K_N<\infty$ and
\begin{align*}
\E\Big[W^+\1\Big\{W^->z^\frac{1}{(N-1)\delta-1},T\leq(W^-)^{N\delta}\Big\}\Big] &\leq K_N \E\left[W^+\1\left\{W^->z^\frac{1}{(N-1)\delta-1}\right\}(W^-)^{N\delta-t\gamma_\text{c}}\right]\\
&\leq K_N\E\left[(W^-)^{\frac{\beta^--1}{\beta^+-1}}\1\left\{W^->z^\frac{1}{(N-1)\delta-1}\right\}\right]z^\frac{N\delta-t\gamma_\text{c}}{(N-1)\delta-1}\\
&= K_N\frac{\beta^--1}{1-\gamma_\text{c}} z^\frac{\gamma_\text{c}-1+N\delta-t\gamma_\text{c}}{(N-1)\delta-1}= o(z)
\end{align*}
since by the choice of $\delta$ and $N$, it holds $\gamma_\text{c}-1+N\delta-t\gamma_\text{c} < (N-1)\delta-1 < 0$.

Finally, we have to consider the part
$\E\left[W^+\psi_T(W^-z)\1{\{(W^-)^{\gamma_\text{c}-\delta}<T\leq(1+\epsilon)\alpha_\text{c}(W^-)^{\gamma_\text{c}}\}}\right]$.
If we choose $\epsilon>0$ small enough such that $(1+2\epsilon)\alpha_\text{c}<\liminf_{w\to\infty}\tau(w)/w^{\gamma_\text{c}}$ and denote $\tilde{\tau}(w):=(1+\epsilon)\alpha_\text{c}w^{\gamma_\text{c}}$, then we observe that by Assumption \ref{ass:exposures} for $w_i^-$ large enough
\begin{align*}
\P\left(\tau_i\leq\tilde{\tau}(w_i^-)\right) &\leq 
\P\left(\sum_{j=1}^{\tilde{\tau}(w_i^-)} E_{\rho_i(j),i}\geq\frac{1+2\epsilon}{1+\epsilon}\tilde{\tau}(w_i^-)\mu_i\right)
\leq K\left((1+\epsilon)\alpha_\text{c}(w_i^-)^{\gamma_\text{c}}\right)^{-(t-1)}
\end{align*}
for some $K<\infty$. Since $\gamma_\text{c}-1-(t-1)\gamma_\text{c} < \gamma_\text{c}+\delta-1 < 0$, we then get
\begin{align*}
&\E\left[W^+\1{\{W^->z^\frac{1}{\gamma_\text{c}+\delta-1},T\leq(1+\epsilon)\alpha_\text{c}(W^-)^{\gamma_\text{c}}\}}\right] \\
&\leq K \E\left[W^+\1_{\{W^->z^\frac{1}{\gamma_\text{c}+\delta-1}\}}\left((1+\epsilon)\alpha_\text{c}(W^-)^{\gamma_\text{c}}\right)^{-(t-1)}\right]\\
&\leq K \left((1+\epsilon)\alpha_\text{c}\right)^{-(t-1)} \frac{\beta^--1}{1-\gamma_\text{c}}z^\frac{\gamma_\text{c}-1-(t-1)\gamma_\text{c}}{\gamma_\text{c}+\delta-1}= o(z).
\end{align*}

\noindent Altogether, we derive (note that we decomposed the expectation in finitely many summands)
\[ \E\left[W^+\psi_T(W^-z)\1\left\{T\leq(1+\epsilon)\alpha_\text{c}(W^-)^{\gamma_\text{c}}\right\}\right] = o(z) \]
and hence $\limsup_{z\to0+}z^{-1}f(z) = \limsup_{z\to0+}z^{-1}\E\left[W^+\psi_T(W^-z)\right] - 1 < 0$, which shows resilience by Theorem \ref{thm:cont:res}. Note that we can cut off $T$ at $(W^-)^\eta$ for some $\gamma_\text{c}<\eta<1$ to ensure continuity of $d(z)$ by Lemma \ref{lem:cont:differentiability} and Remark \ref{rem:cont:continuous:differentiability}.

Part \ref{cor:threshold:res:2} follows by the same calculations replacing $\gamma_\text{c}$ by $\gamma$ and using $\gamma>\gamma_\text{c}$.
\end{proof}

{\begin{multicols}{2}
\footnotesize
\bibliography{finance}
\bibliographystyle{abbrv}
\end{multicols}}

\end{document}